\documentclass[aoas]{imsart}

\RequirePackage{amsthm,amsmath,amsfonts,amssymb}
\RequirePackage[authoryear]{natbib}
\RequirePackage{amsthm,amsmath,amsfonts,amssymb}
\RequirePackage[colorlinks,citecolor=blue,urlcolor=blue]{hyperref}

\usepackage[ruled,vlined]{algorithm2e}
\usepackage{algpseudocode}
\usepackage{graphicx,floatrow}
\usepackage{enumitem}
\usepackage{subcaption}
\usepackage[dvipsnames]{xcolor}
\usepackage{booktabs}
\usepackage{wrapfig}
\usepackage{fancybox}
\usepackage{url}

\startlocaldefs

\theoremstyle{plain}

\newtheorem{thm}{Theorem}
\newtheorem{Assumption}{Assumption}

\newtheorem{Corollary}{Corollary}

\newtheorem{Example}{Example}

\theoremstyle{remark}
\theoremstyle{plain}
\newtheorem{remark}{Remark}

\newcommand{\X}{\mathbf{X}}

\renewcommand{\S}{\mathbf{S}}
\newcommand{\s}{\mathbf{s}}
\newcommand{\Z}{\mathbf{Z}}
\newcommand{\z}{\mathbf{z}}

\renewcommand{\P}{\mathbb{P}}
\newcommand{\Ytilde}{\widetilde{Y}}

\newcommand{\I}{\mathbb{I}}

\newcommand{\p}{\mathbf{p}}

\newcommand{\series}[1]{\{#1\}_{t\ge0}}
\newcommand{\IID}{{IID}}
\newcommand{\DID}{{DID}}
\DeclareRobustCommand{\boxit}[1]{\scalebox{.8}{\ovalbox{#1}}}

\newcommand{\ra}[1]{\renewcommand{\arraystretch}{#1}}

\endlocaldefs

\begin{document}

\begin{frontmatter}
\title{Detecting Distributional Differences in Labeled Sequence Data with Application to Tropical Cyclone Satellite Imagery}
\runtitle{Distributional Differences in Labeled Sequence Data}

\begin{aug}
\author[A]{\fnms{Trey} \snm{McNeely}}, 
\author[A]{\fnms{Galen} \snm{Vincent}},
\author[B]{\fnms{Kimberly M} \snm{Wood}}, 
\author[C]{\fnms{Rafael} \snm{Izbicki}}, 
\and
\author[A]{\fnms{Ann B} \snm{Lee}}
\address[A]{Department of Statistics and Data Science, Carnegie Mellon University}
\address[B]{Department of Geosciences, Mississippi State University}
\address[C]{Department of Statistics, Federal University of S\~ao Carlos}
\end{aug}

\begin{abstract}
Our goal is to quantify whether, and if so how, spatio-temporal patterns in tropical cyclone (TC) satellite imagery signal an upcoming rapid intensity change event. To address this question, we propose a new nonparametric test of association between a time series of images and a series of binary event labels. We ask whether there is a difference in distribution between (dependent but identically distributed) 24-h sequences of images preceding an event versus a non-event. By rewriting the statistical test as a regression problem, we leverage neural networks to infer modes of structural evolution of TC convection that are representative of the lead-up to rapid intensity change events. Dependencies between nearby sequences are handled by a bootstrap procedure that estimates the marginal distribution of the label series. We prove that type I error control is guaranteed as long as the distribution of the label series is well-estimated, which is made easier by the extensive historical data for binary TC event labels. We show empirical evidence that our proposed method identifies archetypes of infrared imagery associated with elevated rapid intensification risk, typically marked by deep or deepening core convection over time. Such results provide a foundation for improved forecasts of rapid intensification.
\end{abstract}


\end{frontmatter}

\section{Introduction}\label{sec:intro}


A broad array of problems in the physical, environmental and biological sciences feature high-dimensional time series $\{\X_t\}_{t \geq 0}$,  associated with binary ``labels'' $\{Y_t\}_{t \geq 0}$ indicating an event of interest. Examples include sequences of satellite or other remote sensing data paired with natural events like the occurrence of an earthquake or the rapid intensification of a hurricane; or multivariate electroencephalographic (EEG) and magnetoencephalographic (MEG) data showing brain activity paired with physiological events like the occurrence of a stroke. Most research on this front concerns prediction of events \citep{luo2014correlating}, measurement of event impact \citep{scharwachter2020does,scharwachter2020two}, or detection of change-points \citep{aminikhanghahi2017survey,evans2020sequential} after events occur. Furthermore, joint analyses of time and event series often assume that the time series is univariate, or model the relationship between multiple scalar quantities as they change over time. There is a lack of theoretically and computationally sound methods for (non-parametric)  association studies and statistical tests for dependent {\em and} high-dimensional sequence data.

This work is motivated by the need to identify spatio-temporal patterns in the convective evolution of tropical cyclone satellite imagery prior to a rapid intensity change; see ``Motivating Application.'' Our immediate goal is not operational forecasting or prediction per se, but rather gaining scientific insight into the spatio-temporal evolution $\S_{<t}=\{\X_{t-T},\X_{t-T+1}, \ldots,\X_{t}\}$ of convective structure or satellite imagery  $\X_{t}$, leading up to a rapid intensity change event ($Y_t=1$), for some lead time $T,$ and identifying whether it differs in distribution from sequences  $\S_{<t}$ that precede a non-event ($Y_t=0$).

From a statistical methodology standpoint, this problem amounts to a challenging two-sample testing problem for high-dimensional {\em dependent but identically distributed} (\DID) data. From the observed time and event series, we extract labeled sequence data $\{(\S_{<t}, Y_t)\}_{t \geq 0}$, which we assume is a stationary process. That is, both $\S_{<t}$ and $Y_t$ are auto-correlated and dependent for different instances of time $t$. By the assumption of stationarity, the data  $(\S_{<t}, Y_t)$ are identically distributed. Given historical data,  we test whether the distributions of $\S_{<t}|Y_t=1$ and   $\S_{<t}|Y_t=0$ are the same. The challenge is to construct an efficient test of (\ref{eqn:original_null}) that is valid (controls type I error) for \DID\ data, and that applies to different types of sequence data (images, functions, and sequentially observed data from multiple physical probes) {with a minimum of assumptions}. Finally, for many problems of applied interest, scientists want to know not just whether sequences preceding an event versus non-event are significantly different in distribution, but if so, also {\em how} the two distributions are different. That is, if the null hypothesis that the distributions of (stationary) sequences $\S | Y=1$ and $\S | Y=0$ are the same is rejected, the question is how to identify the patterns in the state space $\mathcal{S}$ of  $\S$ that contributed to the rejection.\footnote{The stationarity assumption makes the state space $\mathcal{S}$  well-defined, and allows us to drop the subindex $<t$ in our notation.} These patterns correspond to sequences $\s \in \mathcal{S}$ that are more or less likely to be associated with an event ($Y=1$) than by chance.\\

\noindent \textbf{Motivating Application: Tropical Cyclones.} Tropical cyclones (TCs) are highly structured storms which rank among the deadliest and costliest natural disasters in the United States \citep{Klotzbach2018}. Cases of rapid intensification (RI) and rapid weakening (RW) of such storms --- defined for this work as a change in maximum wind speed of at least 25 knots within 24 hours{, $Y=1$} --- are notoriously difficult to forecast \citep{kaplan2003large,Kaplan2010,kaplan2015rii,wood2015definition}. RI prediction has thus been the ``highest-priority forecast challenge'' identified by the National Hurricane Center (NHC) in the last decade \citep{gall2013hurricane}. RW events, while a lower priority than RI, are also associated with above-average forecast errors and are of great interest to meteorologists.

Models such as SHIPS-RII \citep{kaplan2015rii} have made great progress on skillfully forecasting RI events, but these approaches rely on scalar predictors (e.g., area-averaged vertical wind shear) at fixed points in time and thus neglect the evolving spatial structure of the TC, and these structural changes often influence such events. To address this gap, meteorologists and forecasters seek interpretable patterns in the spatiotemporal structure of physically-relevant 2D fields that could indicate an elevated risk of RI. The first step in this search is to find interpretable temporally-evolving sequences of spatial structure $\S_{<t}$ that differ in distribution depending on whether a TC is undergoing a rapid intensity change event ($Y_t=1$) or not ($Y_t=0$). 

Our application examines one such 2D field: deep convection within the storm. Convection, or deep thunderstorm-like clouds, is a key component of the mechanism through which TCs extract energy from the ocean, meaning that convection, in both strength and distribution, should be closely related to storm intensity. We quantify convective structure using cloud-top temperature as measured by infrared (IR) imagery from the Geostationary Operational Environmental Satellites (GOES). As convection strengthens, the cloud tops are pushed higher into the atmosphere where temperatures are lower. The temperature of the cloud top can therefore be used as a proxy for the strength of convection in the storm. We ask: ``Do 24-hour sequences of convective structure, $\series{\S_{<t}}$, contain information about upcoming intensity change, $\series{Y_t}$? If so, how do the spatio-temporal patterns of rapidly changing TCs differ from that of not rapidly changing TCs?"

Both $\series{\S_{<t}}$ and $\series{Y_t}$ are highly dependent (auto-correlated) time series in this application: environmental fields such as convection change slowly, and rapid intensity change events are by definition extended periods of change (typically 12-48 hours), meaning that successive measurements of these variables at short time intervals (e.g., 3 hours) will be highly dependent. It is particularly challenging to perform traditional two-sample tests in this \DID\ setting due to the combination of low sample size (671 TCs between 2000-2020), high-dimensional image data, and strong temporal dependence.\\

\noindent \textbf{Contribution and Relevance.} Our contribution is two-fold: (i) On the methodology side,  we present a new statistical framework for detecting arbitrary distributional differences in a high-dimensional setting with labeled sequence data. The proposed two-sample test is valid in a \DID\ setting and provides local diagnostics as to the type of sequences $\s \in \mathcal{S}$ that contribute to rejection of $H_0$ in Equation~\ref{eqn:new_null}. (ii) On the applied side, we utilize our proposed framework to identify and describe patterns of convective evolution in TCs prior to the onset of rapid intensity change.

\begin{wrapfigure}{R}{0.6\linewidth}
  \centering
  \vspace{-2.75em}
  \includegraphics[width=\linewidth]{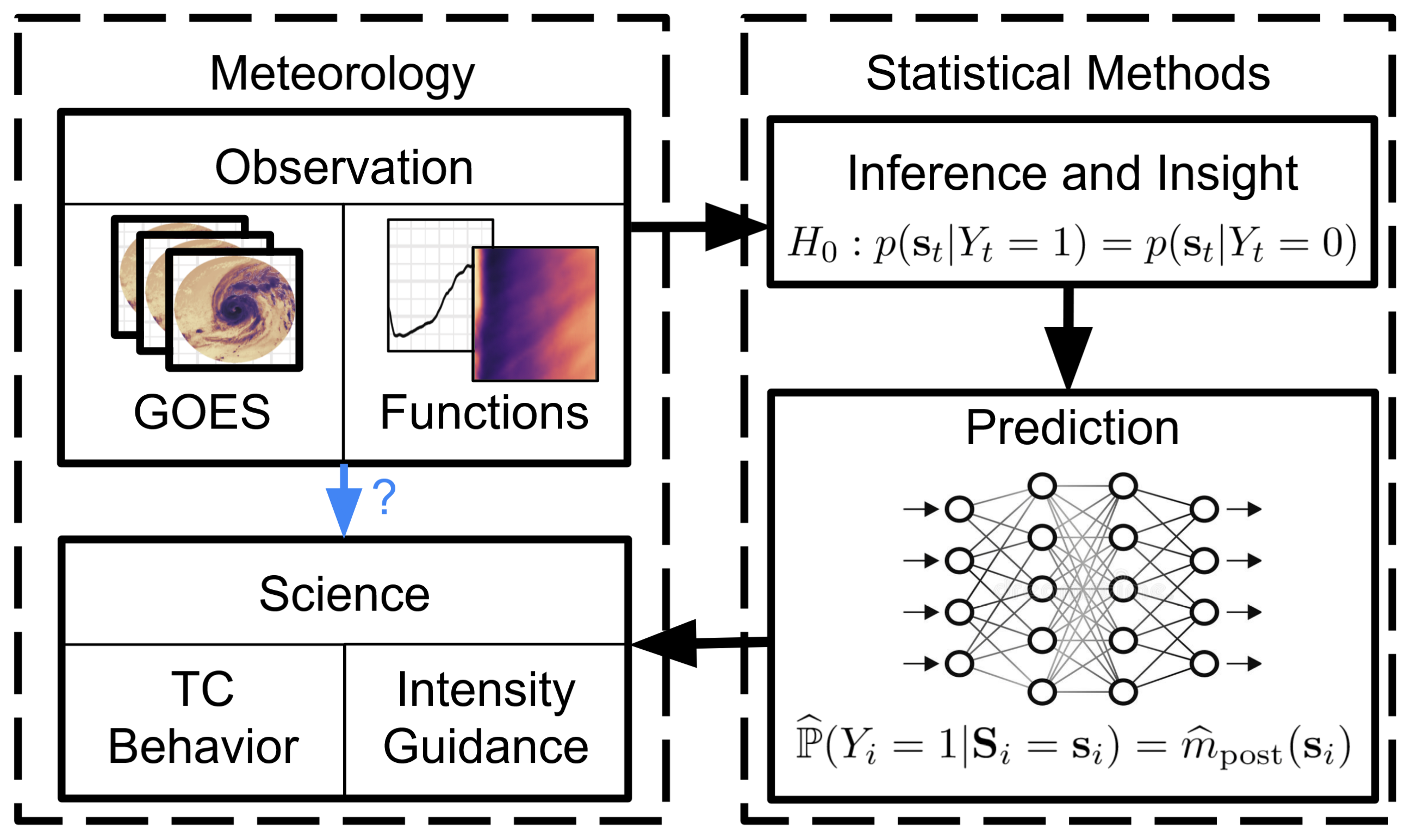}
  
  \caption{\small \textbf{Leveraging prediction tools for statistical inference and scientific insight.}  In our application, we seek relationships between high-dimensional TC observations and TC behavior. We pose this as a two-sample testing problem, but the high dimensionality and sequential correlation in our data make this hypothesis difficult to test. By rewriting the test in terms of a prediction problem, we are able to leverage powerful prediction methods to infer modes of structural evolution in TCs which indicate rapid intensity change.}
  \label{fig:flowchart}
\end{wrapfigure}
 
Figure \ref{fig:flowchart} shows a flow chart of our overall approach. As indicated by the vertical blue arrow to the left, we seek quantitative conclusions regarding how TC behavior relates to convective structure as observed by GOES imagery and extracted functions. We cast the two-sample test as a prediction problem  (Figure \ref{fig:flowchart}, right, for ``Statistical Methods''). This allows us to leverage powerful prediction techniques, such as convolutional neural nets (CNNs), to gain scientific insight from high-dimensional functional or video data without a prior dimension reduction (arrow back to ``TC behavior'' and ``Intensity Guidance''). To account for dependence in the labeled sequence $\series{(\S_{<t},Y_t)}$, we develop a bootstrap regression test (Algorithm \ref{alg:distributional-differences}), which yields a valid p-value accompanied by interpretable diagnostics.

Our bootstrap test is ideal for TC studies: the low number of unique TCs for which high-resolution satellite imagery is available prohibits efficient inference via classical blocking schemes, whereas our method can take advantage of the extensive historical record of event (label) sequences. Theorem \ref{thm:convergence_test_stat} shows that the bootstrap test is valid as long as we can estimate the distribution of the label sequences well; this is made easier by the fact that labels are binary. 

Section \ref{sec:results_TC} provides evidence that deep and/or deepening convection is a necessary precursor to RI, but that other factors (e.g., low vertical shear) must be present. This elevated risk of RI due to deepening convection is often present \emph{prior} to the onset of intensification, demonstrating value to forecasting pipelines. RW, meanwhile, is a more variable process and did not return significant results.\\

\noindent \textbf{Relation to Other Work.}
{\em Event impact and causal inference for time series.} Our problem set-up is closest in flavor to association and causal inference studies for testing the relationship between a time and event series \citep{scharwachter2020two,luo2014correlating}. The vast majority of these works assume univariate time series, or test for each dimension in a multivariate time series separately \citep{candes2018panning}. The recent paper by \cite{scharwachter2020two} leverages a two-sample testing approach for high-dimensional data as in this work, albeit to study how a discrete binary event history impacts a multivariate time series, rather than how the evolution of a complex time series is associated with a later event. A key methodological difference is that their work handles the lack of independence by sampling data points that are far from each other, while our method allows for the use of all available data. This distinction is key for applications with limited data. Their proposed multiple testing procedure also does not control the false positive error rate exactly, but heuristically.
 
{\em Two-sample tests in high dimension.} Recently, there has been a growing interest in nonparametric two-sample tests in high dimension. Popular machine learning-based approaches include classification accuracy tests \citep{kim2021classification}, kernel-based tests \citep{gretton2012kernel}, and divergence-based density ratio tests \citep{Moon2014,kandasamy2015}. We use the same regression test statistic (Equation \ref{eqn:test_stat}) as \cite{kim2019global} to allow for interpretable local diagnostics. However, \cite{kim2019global} and the above-mentioned two-sample testing papers only handle the standard independent, identically distributed (\IID) data setting (Figure \ref{fig:setting1}), whereas the methods in this work apply to \DID\ sequence data (Figure \ref{fig:setting3}). 

Also related to our paper are modern tests of conditional independence between two random vectors $Y$ and $Z$ given a third random vector $X$ (see Discussion in Section \ref{sec:confounding}) and tests of the conditional mean and quantile dependence of $Y$ on $X$. Most high-dimensional research on this front (including model-X knockoffs; \citealt{candes2018panning,sesia2019,berrett2020conditional}) assumes \IID\ data $(X_{i1}, \ldots, X_{ip}, Y_i) \sim F_{X,Y}$ for $i=1,\ldots,n$, and also assumes that the distribution of $Y$ given $X$ depends on only a small fraction of the $p$ covariates. 
The latter sparsity assumption is reasonable for, e.g., genome-wide association studies, but not for remote sensing applications with image and functional data. Another key difference between our testing approach and so-called model-X methodologies, which also use machine learning algorithms to approximate the distribution of $Y$ given $X$ \citep{katsevich2020theoretical}, is that model-X approaches estimate or make assumptions regarding the distribution of $X$ (or $X$ given $Z$), whereas we instead estimate the distribution of the response $Y$.

{\em Bootstrap for time series.} There exist many different types of bootstrap methods for dependent data; see, e.g, \cite{buhlmann2002bootstraps, horowitz2003bootstrap, kreiss2011bootstrap} for a review. The goal is  often to model the data distribution for parameter estimation, rather than as here to test for an association between a label series and a high-dimensional time series. Our framework also does not bootstrap the entire distribution of $\{(\S_{<t}, Y_t)\}_{t \geq 0}$ but only the distribution of labels $\{Y_t\}_{t \geq 0}$. For binary labels, this is a much easier estimation problem than, for example, block-bootstrap of high-dimensional time series.

{\em TC analyses.} 
Many TC analysis tools incorporate information from 2D fields via scalar values such as area averages; for example, the operationally-used SHIPS and SHIPS-RII forecast schemes include scalars for the fraction of pixels with IR temperatures below -30$^\circ$C within the 50-200-km annulus \citep{DeMaria1999,kaplan2015rii}. Such approaches discard complex, time-evolving structure in the 2D fields. More recent analyses take spatial information into account by applying dimension reduction techniques like functional principal component analysis (PCA) to the field, such as for TC eye formation forecasts in \citep{knaff2017forecasting}. However, dimension reduction adds an extra layer of abstraction between TC structure and subsequent TC behavior and can reduce meteorologists' ability to interpret the information. The {ORB} framework (Organization, Radial structure, Bulk morphology) was introduced in \cite{mcneely2020unlocking} to summarize convective structure into a dictionary of functional features. The key objective was to utilize entire functions to quantify structure rather than thresholded feature statistics, thereby enabling richer descriptions of spatial structure while remaining interpretable. In this paper, we leverage CNNs to model the relationship between TC intensity change and the temporal evolution of one such continuous {ORB} function --- the radial profile of cloud-top temperatures as observed by GOES-IR imagery \citep{Sanabia2014,mcneelyquantifying,mcneely2020unlocking}. Our bootstrap test then provides a powerful tool for directly assessing whether there is an association between RI or RW events and TC structure or the TC environment in terms of structural summaries (1D {ORB} functions) that are easily digestible to meteorologists.

\noindent {\bf Outline.} 
We begin by defining the problem set-up in Section \ref{sec:setup}, paying special attention to different dependence structures in $\series{(\S_{<t},Y_t)}$. In Section \ref{sec:tcdata} we describe the TC data. In Section \ref{sec:methods} we lay out the details of our bootstrap test for distributional differences in dependent sequence data. In Section \ref{sec:theory} we provide theoretical justification for validity of the bootstrap test. In Section \ref{sec:synthetic} we introduce a simulated toy example to empirically demonstrate the advantage of a Markov Chain-based bootstrap test over traditional permutation testing. In Section \ref{sec:results_TC} we apply our method to study the evolution of convective structure in TCs prior to a rapid intensity change. Finally, in Section  \ref{sec:extensions} we discuss potential extensions of our method.

\section{Set-Up}\label{sec:setup}
Our goal is to detect distributional differences in labeled sequence data $\{(\S_{<t}, Y_t)\}_{t\ge0}$, where the ``labels'' $Y_t \in \{0,1\}$ are binary, and the covariates $\S_{<t} \in \mathcal{S}$ can represent high-dimensional quantities. We formalize this question in the hypothesis
\begin{align}
\label{eqn:original_null} 
   H_0: p (\s_{<t} | Y_t=1) &= p (\s_{<t} | Y_t=0),\ \text{for all} \ t \ \text{and}\ \s_{<t},   \ \text{versus}\\\ H_1: p (\s_{<t} | Y_t=1) &\neq p (\s_{<t} | Y_t=0), \ \text{for some} \ t \ \text{and}\ \s_{<t}.\notag
\end{align}
We assume that the sequence $\{(\S_{<t},Y_t)\}_{t \geq 0}$ is stationary (Assumption \ref{assump:stationary}), which allows us to rewrite the hypothesis as
\begin{align}
\label{eqn:new_null}
   H_0: p (\s | Y=1) &= p (\s | Y=0),\ \text{for all} \ \s \in \mathcal{S},   \ \text{versus}\\\ H_1: p (\s | Y=1) &\neq p (\s | Y=0), \ \text{for some} \ \s \in \mathcal{S}.\notag
\end{align}

We will study three different settings (see Figure \ref{fig:settings}) which admit increasingly complex dependence in $\series{\S_{<t}}$ and $\series{Y_t}$. In Setting A (Figure \ref{fig:setting1}) there is no temporal dependence, meaning the data $\{(\S_{<t}, Y_t)\}$ are \IID. Testing Equation \ref{eqn:original_null} is still challenging in this setting when $\S_{<t}$ is high-dimensional, but various methods and associated theory have been developed to handle these challenges; see e.g.~\cite{kim2019global,kim2021classification}. In Setting B (Figure \ref{fig:setting2}) there is temporal dependence in $\{\S_{<t}\}$, but $\{Y_t\}$ are conditionally independent given the associated value in $\{\S_{<t}\}$. In Setting C (Figure \ref{fig:setting3}) there is temporal dependence in both $\{\S_{<t}\}$ and $\{Y_t\}$, regardless of the association between the two variables. We expect the TC data to exhibit the structure of Setting C because intensity change labels ($Y_t$) are not conditionally independent solely given the convective structure of the storm ($\S_{<t}$). {The effect of convective activity on TC intensity change generally manifests within 24 hours, so a 24-hour history in $\S_{<t}$ should be sufficient \citep{rogers2010convective}.}

\begin{figure}[htbp]
	\centering
	
	\begin{subfigure}[t]{0.3\linewidth}
		\centering
		\vskip 0pt
		\includegraphics[width=\linewidth]{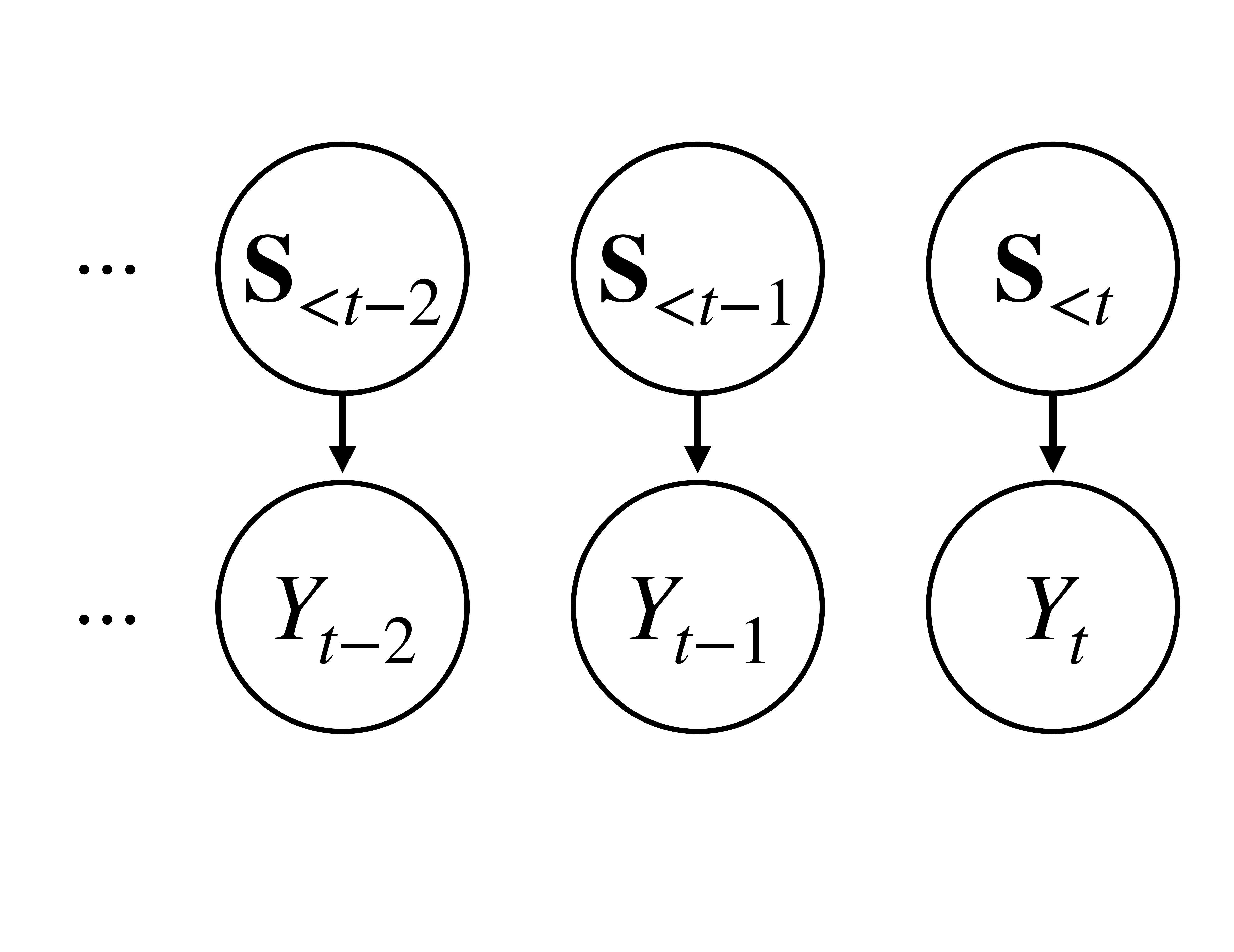}
		\caption{\small Setting A: $\{(\S_{<t}, Y_t)\}_{t\ge0}$ with no temporal dependence between pairs $(\S_{<t},Y_t)$ for different $t$.}
		\label{fig:setting1}		
	\end{subfigure}
	\quad
	\begin{subfigure}[t]{0.3\linewidth}
		\centering
		\vskip 0pt
		\includegraphics[width=\linewidth]{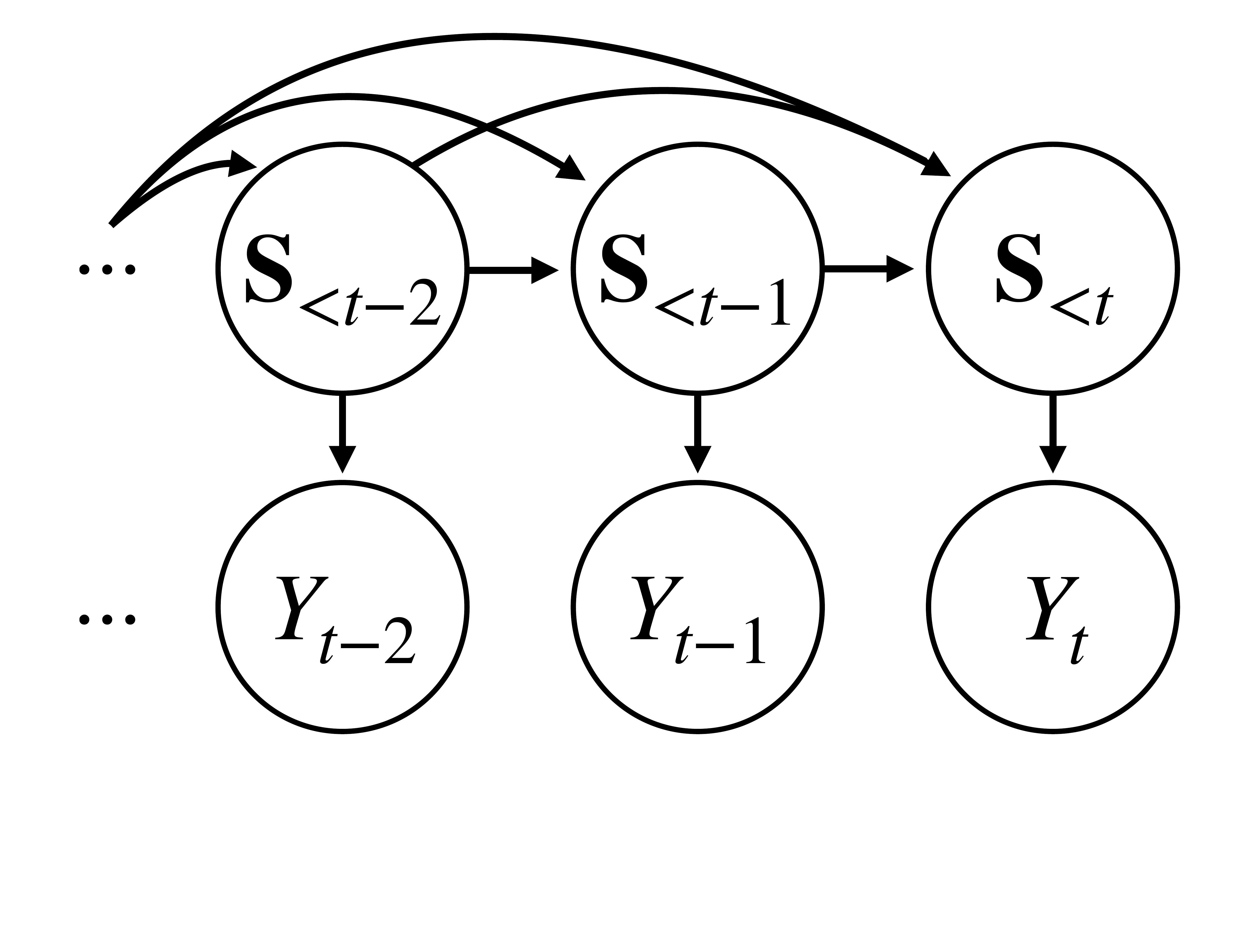}
		\caption{\small Setting B: $Y_t$ conditionally independent of $Y_{t-1}$ given $\S_{<t}$; $\S_{<t}$ is autocorrelated.} 
		\label{fig:setting2}		
	\end{subfigure}
	\quad
	\begin{subfigure}[t]{0.3\linewidth}
		\centering
		\vskip 0pt
		\includegraphics[width=\linewidth]{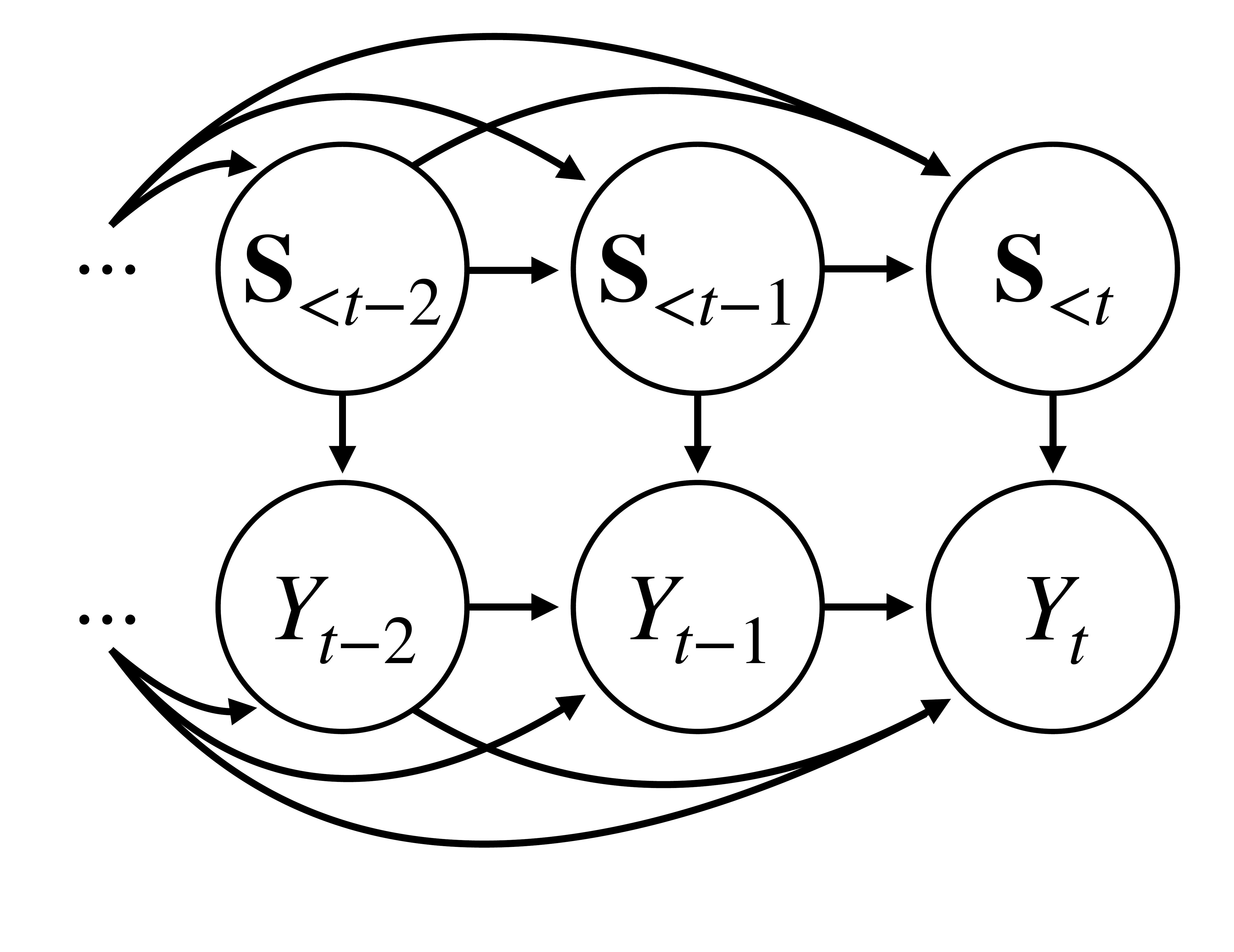}
		\caption{\small Setting C: $Y_t$ conditionally dependent on $Y_{t-1}$ given $\S_{<t}$; $\S_{<t}$ and $Y_t$ are each autocorrelated.}
		\label{fig:setting3}		
	\end{subfigure}
	\caption{\small \textbf{Dependence settings.} Directed acyclic graphs (DAGs) illustrating the three dependence structures we explore. Note that each variable $\S_{<t}$ can itself represent a temporal sequence of high-dimensional functions or images, as in Figure \ref{fig:profiles}.}\label{fig:settings}
\end{figure}

\section{Sequence Data from Tropical Cyclone Satellite Imagery}\label{sec:tcdata}
\begin{figure}[!htbp]
    \centering
    \includegraphics[width=\linewidth]{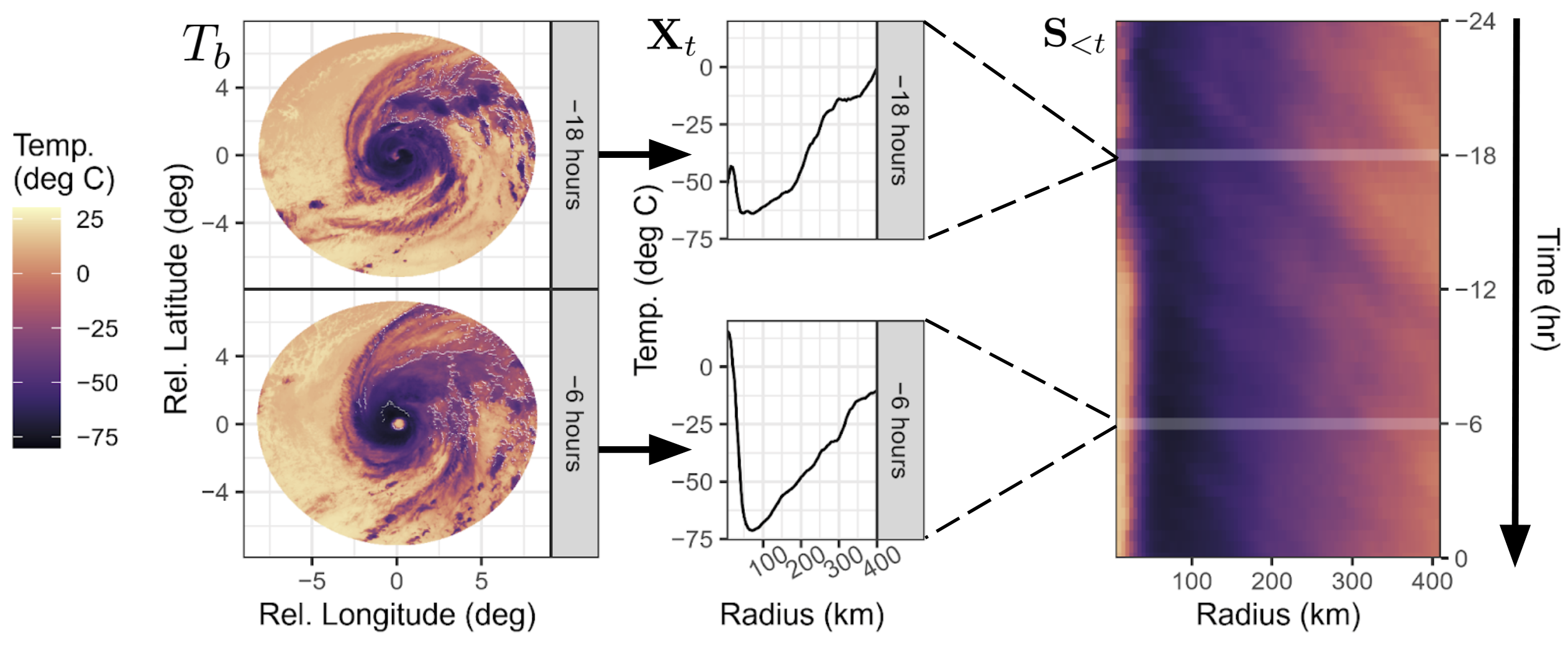}
    
    \caption{\small  \textbf{Evolution of TC convection as structural trajectories.} The raw data for each trajectory $\S_{<t}$ is a sequence of a sequence of TC-centered cloud-top temperature images from GOES ($T_b$). We convert each GOES image into a radial profile ($X_t$). The 24-hour sequence of consecutive radial profiles, sampled every 30 minutes, defines a \emph{structural trajectory} or Hovm\"oller diagram ($\S_{<t}$). These trajectories serve as high-dimensional inputs to  $\widehat{m}_\text{post}(\s_{<t})$.}
    \label{fig:profiles}
\end{figure}

\begin{table}[!htb]
\ra{1.2}
\centering
\begin{tabular}{@{}rl|rrr|rl@{}}\toprule
   &                                & NAL    & ENP  & Total  & Year Range & Years      \\ \hline
(i)& \textbf{Training Data}     &        &      &        &             &            \\
&All Sequences                  & 47,502 &31,549& 79,051 &             &            \\
&RI  Sequences                  &  7,015 & 6,742& 13,757 &             &            \\
&RW Sequences                   &  5,878 & 7,298& 13,176  &           &            \\
&\textbf{Unique TCs}            &  209   & 185  &  394   & 2000-2012   & \textbf{13}\\ \hline
(ii)& \textbf{Test Data}        &        &      &        &             &            \\
&All Sequences                  & 28,368 &32,817& 61,185 &             &            \\
&RI Sequences                   &  3,965 &6,386 & 10,351  &             &            \\
&RW Sequences                   &  3,167 &7,182 & 10,349  &             &            \\
&\textbf{Unique TCs}            &    125 & 152  &  277   &2013-2020    & \textbf{8} \\ \hline
(iii)& \textbf{Synoptic Labels}     &        &      &        &             &            \\ 
&All Labels                         & 14,683 &15,274& 29,957 &             &            \\
&RI Labels                          &  1,850 & 2,462& 4,312  &             &            \\
&RW Labels                          &  1,643 & 2,534&  4,177 &             &            \\
&\textbf{Unique TCs}                &  532   &  589 &  1,121 & 1979-2012   & \textbf{34}\\
\bottomrule
\end{tabular}
\caption{\small \textbf{Sample sizes:} Data set summary for each category, (i) labeled Sequences $(\S_{<t}, Y_t)$ used in training, (ii) unlabeled test sequences $\S_{<t}$, and (iii) synoptic labels $Y_t$  used when complete trajectories are not needed.}
\label{table:sample}
\end{table}

Analysis of TC convective structure relies on two types of observations: sequences of longwave infrared imagery captured by GOES imagers and records of TC intensity and location recorded in NOAA's HURDAT2 database \citep{Landsea2013}. 

Longwave infrared (IR) imagery ($\sim$10.3$\mu m$ wavelength) serves as a proxy for convective strength: where IR-estimated cloud-top temperatures are low, convection is strong. GOES longwave IR imagery is available through NOAA's MERGIR database \citep{mergir} at 30-minute$\times$4-km resolution over both the North Atlantic (NAL) and Eastern North Pacific (ENP) basins from 2000--present. Every thirty minutes during the lifetime of a storm, we download a $\sim$2,000 km$\times$2,000 km ``stamp'' of IR imagery surrounding the TC location. Figure \ref{fig:profiles} (left) shows two such stamps after an 800-km radius mask is applied.

TC location and intensity are given by the NHC's HURDAT2 best track database, which utilizes all available data on each TC (including data not available in real time) to estimate critical characteristics over the lifetime of each TC. HURDAT2 best tracks are provided at 6-hour (or synoptic) time resolution; we linearly interpolate TC latitude and longitude between HURDAT2 data points to estimate TC location at non-synoptic times. Since we are interested in the behavior of mature TCs (as opposed to e.g. early development), we consider TC genesis to be the first synoptic time at which intensity surpasses 35 kt and lysis to be the last synoptic time at which intensity is at least 35 kt.

\textbf{Structural trajectories via {ORB}.}
Direct utilization of IR imagery is challenging in a meteorological setting, where interpretability is at a premium. In previous works, we defined the {ORB} framework (Organization, Radial structure, Bulk morphology) to summarize convective structure into a dictionary of continuous functions instead of using traditional area-averaged features common to TC analyses. This framework enables richer descriptions of spatial structure without sacrificing interpretability \citep{mcneelyquantifying}. We have leveraged the {ORB} framework to analyze the evolution of TC convective structure and demonstrated how projections of ORB functions onto a PCA basis can be used to identify rapid intensification events \citep{mcneely2020unlocking}, but we have not yet directly utilized temporally evolving, continuous functions.

This work studies the temporal evolution of an entire {ORB} function; in this case, the radial profile, defined as $\overline{T}(r)=\frac{1}{2\pi}\int_0^{2\pi}T_b(r,\theta)d\theta$. The radial profile $\overline{T}(r)$ captures the structure of cloud-top temperatures $T_b$ as a function of radius $r$ from the TC center and serves as an easily interpretable description of the depth and location of convection near the TC core \citep{Sanabia2014,mcneely2020unlocking}. The radial profiles are computed at 5-km resolution from 0-400km ($d=80$) (Figure \ref{fig:profiles}, center); we denote these summaries of convective structure at each time $t$ by $\X_t$. Finally, at each time $t$ we stack the preceding 24 hours (48 profiles) into a \emph{structural trajectory} $\S_{<t}=\{\X_t,\X_{t-1},\dots,\X_{t-48}\}$. We visualize these trajectories with Hovm\"oller diagrams (\citealt{hovmoller1949trough}; see Figure \ref{fig:profiles}, right).

\textbf{Labeling TC sequence data.} HURDAT2 contains estimated TC intensities only at synoptic times (0000 UTC, 0600 UTC, 1200 UTC, and 1800 UTC). We thus begin by labeling these points $Y_t\in\{0,1\}$ based on whether the TC was undergoing RI (or RW, for those analyses) at time $t$, where $Y = 1$ indicates occurrence of a rapid intensity change event. We then interpolate to non-synoptic times by assigning label $Y_t=1$ if the an observation falls between two consecutive synoptic $Y_t=1$ observations, and $Y_t=0$ otherwise. See \cite{mcneely2022Supplement} (Section B) for further details on this procedure.

There are three sample sizes of interest in this application: (i) the number of labeled training sequences $\S_{<t}$ (further divided into 60\% train/40\% validation), (ii) the number of test sequences $\S_{<t}$, and (iii) the number of synoptic best track entries used when only labels $Y_t$ are required (e.g., $\widehat{m}_\text{seq}$ in Algorithm  \ref{alg:distributional-differences}(3)). These sample sizes and associated years are summarized in Table \ref{table:sample}.

\section{Methods}\label{sec:methods}
Our TC problem set-up is difficult because of: (i) the complexity of the data themselves, with one observation representing an entire sequence $\S_{<t}$ of functions; (ii) dependence between labels $Y_t$ (and sequences $\S_{<t}$) at nearby time points $t$; and finally (iii) the need for scientific interpretability, or more precisely, statistical findings which are easily digestible by TC scientists and forecasters.

We aim to test hypothesis~\ref{eqn:original_null}  for \DID\ sequence data $\series{(\S_{<t}, Y_t)}$ that satisfy the DAG of Setting~C, Figure~\ref{fig:setting3}. Our approach builds on \cite{kim2019global}, where the authors present a regression approach for detecting differences in high-dimensional \IID\ data $\{(\S_i, Y_i)\}_{i=1}^{n}$, where $\S_i \in \mathcal{S}:=\{\s \in \mathbb{R}^D: p(\s)>0\}$, for $Y_i \in \{0,1\}$. Their set-up is equivalent to Setting~A, Figure~\ref{fig:setting1}. The main idea is to rewrite the two-sample test in the equivalent formulation
\begin{align} 
H_0: \mathbb{P} (Y=1 | \mathbf{S}=\mathbf{s}) &= \mathbb{P} (Y=1), \text{ for all } \mathbf{s} \in \mathcal{S},  \text{ and }\label{eqn:global_hypothesis}\\
H_1: \mathbb{P} (Y=1 | \mathbf{S}=\mathbf{s}) &\neq \mathbb{P} (Y=1), \text{ for some } \mathbf{s} \in \mathcal{S}.\notag
\end{align}
These hypotheses involve a regression function for the ``class posterior'' $m_{\text{post}}(\s):= \mathbb{P} (Y=1 | \mathbf{S}=\mathbf{s})$, and a ``class prior'' $m_{\text{prior}}:= \mathbb{P} (Y=1)$. One then tests $H_0$ against $H_1$ using the test statistic
\begin{equation}
\lambda=\sum_{\s \in \mathcal{V}}\left(\widehat m_{\text{post}}(\s)-\widehat m_{\text{prior}}\right)^2,\label{eqn:test_stat}
\end{equation}
where $\widehat m_{\text{post}}(\s)$ is an estimate of $m_{\text{post}}(\s)$, $\widehat m_{\text{prior}}=\frac{1}{n} \sum_{i =1}^n I(Y_i=1)$ is the class proportion of the {training} sample, and $\mathcal{V} \subset \mathcal{S}$ is a fixed finite set of evaluation data points. Depending on the choice of regression method, the regression test based on $\lambda$ can adapt to challenging non-standard data, like images and sequences of images or functions.

Because the null distribution of $\lambda$ is typically unknown, \cite{kim2019global} compute p-values based on $\lambda$ by using a permutation procedure. The procedure relies on the exchangeability of the labels $Y$ under $H_0$ in Equation~\ref{eqn:global_hypothesis}. However, there are at least two types of dependence in our data $\series{(\S_{<t},Y_t)}$ which violate the assumption of exchangeability: (i) autocorrelation in $\series{Y_t}$ which is inherent or governed by unobserved quantities as in Setting C, Figure~\ref{fig:setting3}, and (ii) the presence of an observed, correlated confounding sequence $\series{\Z_t}$ (discussed in Section \ref{sec:confounding}). In either case, the theoretical guarantees of a valid test in \cite{kim2019global} no longer hold.

\subsection{Accounting for Dependence in $Y_t|\S_{<t}$}
\label{seq:accounting_for_dependencies} How do we handle dependence in the labeled sequence data $\series{(\S_{<t},Y_t)}$?
Permutation tests essentially model the distribution of $\series{Y_t}$ assuming \IID\ labels. One way to admit dependence in the relabeling procedure is to instead assume a Markov property of order $k$ on $\series{Y_t}$; that is, to assume that the random variable $Y_t$ depends only on the previous $k$ variables $Y_{t-1},\dots,Y_{t-k}$. To estimate the null distribution of the test statistic $\lambda$ (\ref{eqn:global_hypothesis}), we draw new labels from a Markov autoregressive model:
\begin{equation}
    \Ytilde_t \sim \text{Bernoulli}\left(\widehat{\P}(Y_t=1\vert Y_{t-1}=\Ytilde_{t-1},\dots,Y_{t-k}=\Ytilde_{t-k})\right).\label{eqn:markov-marginal}
\end{equation}
The marginal distribution of the labels, denoted by $m_\text{seq}$, can be estimated from a holdout sample of observed data $\series{Y_t}$ with a variety of methods, including binary Markov chains and random forests. As we shall see, as long as the marginal estimate of $\{Y_t\}_{t\geq 0}$ converges in distribution to the true data-generating process as the size of the holdout sample increases, then the bootstrap test detailed in Algorithm \ref{alg:distributional-differences} will be asymptotically valid. The result holds even if $m_{\text{post}}(\s)$ and $m_{\text{prior}}$ are not well-estimated. This is good news for TC analysis, as one usually has ample access to label series data $Y_t$, whereas sample sizes for the sequences $\S_{<t}$ derived from high-resolution satellite images are smaller.

Proof of the validity of our bootstrap test for Equation \ref{eqn:global_hypothesis} 
is given in Section \ref{sec:theory}, Theorem \ref{thm:convergence_test_stat}. Section~\ref{sec:synthetic} includes empirical results on the power of the test for synthetic data with the DAG structures in Figure~\ref{fig:settings}.

\begin{algorithm}[!b] 
	\caption{Test for distributional differences in labeled sequence data.}
	\label{alg:distributional-differences}
	{\small
		\begin{algorithmic} 
			\Require type of test (\texttt{BST}=TRUE for bootstrap test; \texttt{BST}=FALSE for permutation test); train data $\{(\S_{<t}, Y_t)\}_{t \in \mathcal{T}_1}$ and regression method for estimating $m_\text{post}(\s):= \mathbb{P} (Y_t=1 | \mathbf{S}_{<t}=\mathbf{s})$; for bootstrap test: train data $\{Y_t\}_{t \in \mathcal{T}_2}$ and regression method for estimating $m_\text{seq}(Y_{t-1},\dots,Y_{t-k}):=\P(Y_t=1|Y_{t-1},\dots,Y_{t-k})$; number of repetitions $B$; evaluation points $\mathcal{V}$.
			
			\Ensure p-value for testing $H_0: p (\s | Y=1) = p (\s | Y=0) \ \text{for all} \ s\in \mathcal{S}$; local posterior differences $\{\lambda(\s)\}_{\s \in\mathcal{V}}$ at the evaluation points $\s \in \mathcal{S}$.
			\vspace{0.7em}\\
			
			\hspace{-2ex}\tcp{\texttt{Estimate underlying probability distributions}}\vspace{0.25em}
			
			\State (1) Estimate $m_\text{prior} := \P(Y_t = 1)$ with 
			class proportion $\widehat{m}_\text{prior}=\frac{1}{\lvert\mathcal{T}_1\rvert} \sum_{t\in\mathcal{T}_1}Y_t$.\vspace{0.25em}
			
			\State (2) Regress $Y_t$ on $\S_{<t}$ using $\mathcal{T}_1$ to compute $\widehat{m}_\text{post}$. \vspace{0.25em}
			
			\State (3) \If{\texttt{BST}}{Regress $Y_t$ on $Y_{t-1},\dots,Y_{t-k}$ using $\mathcal{T}_2$ 
				to compute $\widehat{m}_\text{seq}$.}
			
			\hspace{-2ex}\tcp{\texttt{Compute test statistic and estimate its null distribution}}	\vspace{0.25em}
			
			\State (4) Compute test statistic $\mathbf{\lambda}=\sum_{\s \in\mathcal{V}}\lambda^2(\s),\text{ where }\lambda(\s)=\widehat{m}_\text{post}(\s)-\widehat{m}_\text{prior}$.\vspace{0.25em}
			
			\State (5) \For{$b\in \{1,2,\dots,B\}$}{
				\hspace{2ex}\State - Draw new train labels $\{\widetilde{Y}_t\}_{t\in\mathcal{T}_1}$ under $H_0$: 
				
				\hspace{-1ex}\eIf{\texttt{BST}}{
					\State Draw an initial label sequence $\widetilde{Y}_1,\dots,\widetilde{Y}_k$ from the empirical distribution.
					\State Draw sequence of length 100$\times k$ from $\widetilde{Y}_t\sim\text{Binom}(\widehat{m}_\text{seq}(\widetilde{Y}_{t-k},\dots,\widetilde{Y}_{t-1}))$ for burn-in.
					\State Draw new labels $\widetilde{Y}_t\sim\text{Binom}(\widehat{m}_\text{seq}(\widetilde{Y}_{t-k},\dots,\widetilde{Y}_{t-1}))$ for $t \in \mathcal{T}_1$.
				}{  \State Permute original labels $\{Y_t\}_{t \in \mathcal{T}_1}$.}
				
				\State - Regress $\widetilde{Y}_t$ on $\S_{<t}$ using $\mathcal{T}_1$ to compute $\widehat{m}_\text{post}^{(b)}$.\vspace{0.25em}
				
				\State - Recompute test statistic $\widetilde{\lambda}^{(b)}=\sum_{\s \in \mathcal{V}}\left(\widetilde{\lambda}^{(b)}(\s)\right)^2, \text{ where }\widetilde{\lambda}^{(b)}(\s)= \widehat{m}^{(b)}_\text{post}(\s)-\widehat{m}_\text{prior}$. 
			}
			
			\hspace{-2ex}\tcp{\texttt{Compute approximate p-value}}	\vspace{0.25em}
			
			\State (6) Compute $p$-value according to 
			{$$\widehat{p} = \frac{1}{B+1} \left( 1 + \sum_{b=1}^B \I\left(\widetilde{\lambda}^{(b)} > \lambda\right) \right).$$}
			
			\State \textbf{return} $\widehat{p},\  \{\lambda(\s)\}_{\s \in\mathcal{V}}$
			
		\end{algorithmic}
	}
\end{algorithm}

\subsection{Local Diagnostics}
Suppose $H_0$ is rejected. That is, we detect that the two distributions $p(\s|Y=0)$ and $p(\s|Y=1)$ are indeed different. How do we then provide the scientist with interpretable diagnostics that explain how the two distributions are different?
 
Classification accuracy tests \citep{kim2021classification} require a separate post-hoc procedure to identify local distributional differences \citep{gretton2012kernel,chakravarti2021model}. A key advantage of the regression test is that the test statistic in Equation \ref{eqn:test_stat}, by construction, is a sum of local posterior differences. Indeed,
for each evaluation point $\s \in \mathcal{V}$, we compute the {\em local posterior difference} (LPD) 
\begin{equation}
    \lambda(\s)=\widehat{m}_\text{post}(\s)-\widehat{m}_\text{prior}. \label{eqn:local-importance}
\end{equation}
A large value of $\lambda(\s)$ indicates that the distributions  $p(\s|Y=0)$ and $p(\s|Y=1)$ are very different at $\s \in \mathcal{V}$, which in turn contributes to a larger test statistic $\lambda=\sum_{\s \in \mathcal{V}}\lambda^2(\s)$ and potential rejection of $H_0$.

\begin{remark}
The posterior difference $\lambda(\s)=\P(Y=1|\s)-\P(Y=1)$ can be viewed as a {\em scaled density difference}: 
\begin{equation}\lambda(\s) = \frac{p(\s|Y=1)-p(\s|Y=0)}{w(\s)}
\label{eq:scaled_density_difference}
\end{equation}
where ${w(\s)=\frac{1}{1-\pi} \ p(\s|Y=1) + \frac{1}{\pi} \ p(\s|Y=0)}$ is a positive scaling function, and ${\pi:=\P(Y=1)}$ denotes the prior class probability or ${m}_\text{prior}$. This difference
 has several desired properties for assessing local distributional differences: $\lambda(\s)$ is always bounded, unlike other popular discrepancy measures such as the density ratio, $p(\s|Y=1)/p(\s|Y=0)$, and the density difference,  $p(\s|Y=1)-p(\s|Y=0)$, itself. Furthermore, the posterior difference does not decay to zero as fast as $p(\s)$, which leads to high sensitivity to detect differences in low density regions; for example, in the case of balanced classes,  $\lambda(\s)$ takes a value of $+\frac{1}{2}$ when $p(\s|Y=1)\gg p(\s|Y=0)$, and a value of $-\frac{1}{2}$ when $p(\s|Y=1)\ll p(\s|Y=0)$, regardless of the actual magnitudes of $p(\s)$, $p(\s|Y=1)$, and $p(\s|Y=0)$.  
 \end{remark}

In summary, our method tests for distributional differences in labeled sequence data as follows: 
\begin{enumerate}
    \item Decide on a suitable model for the marginal distribution of $\{Y_t\}_{t\ge0}$.
    \item Apply Algorithm \ref{alg:distributional-differences} to compute the  p-value for testing the hypotheses in Equation~\ref{eqn:global_hypothesis}.
    \item If $H_0$ is rejected, then examine local posterior differences (\ref{eqn:local-importance}) to identify what patterns $\s$ in the state space $\mathcal{S}$ of sequence data contributed the most to the rejection.
\end{enumerate}

\section{Theory}\label{sec:theory}
This section provides  theoretical justification for our bootstrap procedure for testing Equation \ref{eqn:global_hypothesis}.  In particular, we show that Algorithm \ref{alg:distributional-differences} controls the type I error asymptotically.

\begin{Assumption}[\textbf{Stationary sequence}]
\label{assump:stationary}
$\{(\S_{<t},Y_t)\}_{t \geq 0}$ is a stationary sequence, where $\S_{<t} \in \mathcal{S}$ and $Y_t \in \{0,1\}$
\end{Assumption}

Assumption \ref{assump:stationary} is needed  for the hypothesis in Equation \ref{eqn:new_null}
to be well-defined. 

\begin{Assumption}[\textbf{Conditional independence}]
\label{assump:DAG_C}
$\{(\S_{<t},Y_t)\}_{t \geq 0}$ satisfies the DAG of Setting C (Figure \ref{fig:settings}).
\end{Assumption}

Assumption \ref{assump:DAG_C} encodes the conditional independences required for our method to control type I error.

In this section, we denote the test statistic by
\begin{align}
\label{eq:test_statistic}
\lambda(\mathcal{D})=\int (\widehat m_{\text{post}}(\s)-\widehat m_{\text{prior}})^2 dQ(\s),    
\end{align}
where $Q$ is any fixed measure over $\mathcal{S}$,
and  $\widehat m_{\text{post}}$ and $\widehat m_{\text{prior}}$ are obtained using a training set $\mathcal{D}:=\{(\S_{<t},Y_t)\}_{t \in \mathcal{T}_{1}}$. 
In Algorithm~\ref{alg:distributional-differences}, $Q$ is the distribution that assigns mass $1/|\mathcal{V}|$ for each evaluation  point $\s \in \mathcal{V}$, but the results we show here apply to any $Q$.
We also assume that the regression estimator we use is a continuous function of the training set.
\begin{Assumption}[\textbf{Continuous regression method}]
\label{assump:continuos_regression}
$\widehat m_{\text{post}}$ is obtained by applying a  regression estimator that is a continuous function of $\mathcal{D}$.
\end{Assumption}

Moreover, let $\mathcal{D}^{t_2}_0:=\{(\S_{<t},Y^0_t)\}_{t \in \mathcal{T}_{\text{1}}}$ denote a random draw from the data set used to estimate the null distribution of $\lambda$, where 
$\{Y^0_t\}_{t \in \mathcal{T}_{\text{1}}} \sim \text{G}_{\widehat \p_{t_2}}$ and  $\text{G}$ is a
distribution over $\{0,1\}^{|\mathcal{T}_{\text{1}}|}$ indexed by the parameter $\widehat \p_{t_2}$, which is estimated   using a holdout set 
$\mathcal{D}'=\{Y_t\}_{t \in \mathcal{T}_{\text{2}}}$, with $t_2=|\mathcal{T}_{\text{2}}|$. In the method
described in Section \ref{seq:accounting_for_dependencies}, $\text{G}_{\widehat \p_{t_2}}$ is the Markov autoregressive model $\widehat m_{\text{seq}}$, but this model can be more general. We require it to converge to the true distribution of the marginal process $\{Y_t\}_{t \in \mathcal{T}_{1}}$ when the null hypothesis holds:
\begin{Assumption}[\textbf{Consistency of the marginal distribution estimator}]
\label{assump:converge_marginal}
The estimator $\widehat \p_{t_2}$ is such that, if  the null hypothesis is true,
$$\text{G}_{\widehat \p_{t_2}} \xrightarrow[t_2 \longrightarrow\infty]{\enskip \mbox{Dist} \enskip} \text{G}^*,$$
where $\text{G}^*$ is the true generating process of $\{Y_t\}_{t \in \mathcal{T}_{1}}$.
\end{Assumption}

In the following, we show two examples where Assumption \ref{assump:converge_marginal} holds.

\begin{Example}
\label{ex:bernoulli} Under Settings A and B (Figure \ref{fig:settings}), $Y_t$'s
are \IID \ under the null hypothesis. Thus, $G^*$ is necessarily a product of \IID \ Bernoulli random variables with some parameter $p$. Now, let
$G_{\widehat \p_{t_2}}$ be the product of \IID\  Bernoulli random variables with parameter given by $p_{t_2}:=(t_2)^{-1}\sum_{t \in   \mathcal{T}_{2}} Y_t$. The law of large numbers implies that 
 $p_{t_2} \xrightarrow[t_2 \longrightarrow\infty]{\enskip \mbox{a.s.} \enskip} p$.
 Thus, the cumulative distribution function of $Y_t^0$, given by
 $$F_{Y_t^0}(y_t)=\begin{cases}
    0, & \text{if $y_t<0$}\\
    1-p_{t_2}, & \text{if $0\leq y_t<1$}\\
    1 & \text{otherwise}
  \end{cases}$$
 is such that $F_{Y_t^0}(y_t) \xrightarrow[t_2 \longrightarrow\infty]{\enskip \enskip} F_{Y_t}(y_t) $.
 It follows that
 \begin{align*}
    \P\left(Y^0_t  \leq y_t, \  \forall t \in  \mathcal{T}_{1} \right)&=\Pi_{t \in \mathcal{T}_1} F_{Y_t^0}(y_t) \\ &\xrightarrow[t_2 \longrightarrow\infty]{\enskip \enskip} \Pi_{t \in \mathcal{T}_1} F_{Y_t}(y_t) 
     =  \P\left(Y_t  \leq y_t, \  \forall t \in  \mathcal{T}_{1} \right)
\end{align*}
and therefore Assumption \ref{assump:converge_marginal} holds.
\end{Example}

\begin{Example}[Markov Chain] If (under $H_0$) the process $\{Y_t\}_{t \geq 0}$ is an irreducible and ergodic stationary $k$-order Markov Chain, then the maximum likelihood estimators of the transition probabilities converge almost surely to the true transition probabilities \citep{grimmett2020probability}. The same reasoning of Example \ref{ex:bernoulli} then implies that Assumption \ref{assump:converge_marginal} holds for such estimator under Setting C.
\end{Example}

The following theorem shows that, under $H_0$, the test statistic has approximately the 
same distribution as the test statistic evaluated at the 
generated data $\mathcal{D}^{t_2}_0$.
\begin{thm}
\label{thm:convergence_test_stat}
Assume \ref{assump:stationary}, \ref{assump:DAG_C}, \ref{assump:continuos_regression} and \ref{assump:converge_marginal}.
Under the null hypothesis,  
$$\lambda(\mathcal{D}^{t_2}_0) \xrightarrow[t_2 \longrightarrow\infty]{\enskip \mbox{Dist} \enskip} \lambda(\mathcal{D})$$
\end{thm}

It follows from Theorem \ref{thm:convergence_test_stat} that 
type I error is controlled  asymptotically{.}
\begin{Corollary}[\textbf{Type I error control}]
\label{cor:typeI}
Let 
$$\widehat p^{\ t_2}_B(\mathcal{D}):=\frac{1}{B+1} \left(  1+\sum_{b=1}^B \I\left(\lambda (\mathcal{D}^{(b)}) >\lambda(\mathcal{D}) \right) \right)$$ be the Monte Carlo p-value for $H_0$, where $\mathcal{D}^{(1)},\ldots,\mathcal{D}^{(B)} \stackrel{\text{\IID}}{\sim}  \mathcal{D}^{t_2}_0$. Assume that Assumptions \ref{assump:stationary}, \ref{assump:DAG_C}, \ref{assump:continuos_regression} and \ref{assump:converge_marginal} hold. Then, under the null hypothesis, for any $0<\alpha<1$,
$$\lim_{t_2 \longrightarrow \infty}\lim_{B \longrightarrow \infty} \P\left( \widehat p^{\ t_2}_B(\mathcal{D})\leq \alpha \right) =\alpha.$$
\end{Corollary}

See \cite{mcneely2022Supplement} (Section C) for proofs.

\section{Performance of Tests on Synthetic Data}\label{sec:synthetic}
In this section, we use synthetic data to examine the performance (validity, power and diagnostics) of our Markov chain bootstrap test for the data dependence settings in Figure~\ref{fig:setting3}.

\subsection{Synthetic Univariate Sequence Data}\label{sec:synth}
For simplicity, we first consider a scalar covariate $S_t$ of interest. We then create dependent sequences $\{S_t, Y_t\}_{t\ge0}$ with a logistic generative model:
\begin{align}
Y_t|S_t\sim&\,\text{Bernoulli}(p_t)\notag\\
p_t=\,&\text{logistic}(\gamma H_\delta(S_t) + U_t)\notag\\
H_\delta(S)=&
\begin{cases}
0&\lvert S\rvert<\delta\\
S&\lvert S\rvert\ge\delta
\end{cases}\label{eqn:synthetic}\\
S_t =& U_t^\prime, \ \ U_t^\prime \sim  \text{AR}_{\phi'}(1)\notag \\
U_t \sim& \text{AR}_{\phi}(1).\notag
\end{align}

The variables $U_t$ and $U'_t$ are spurious variables (not included in the DAGs), which induce autocorrelation in the binary response variable $Y_t$ and the covariate $S_t$, respectively. In our toy example, we assume that both $U_t$ and $U'_t$ are given by autoregressive models of order $1$; more specifically, by AR(1) models of the form $U_t = \phi U_{t-1} + \sqrt{1-\phi^2}\epsilon_t$ and $U'_t = \phi^\prime U'_{t-1} + \sqrt{1-\phi^{\prime2}}\epsilon'_t$,  where $\epsilon_t, \epsilon'_t \overset{iid}{\sim}N(0,1)$, and $\phi, \phi^\prime {\in [0,1]}$ are parameters for the 1-lag autocorrelation in $U_t$ and $U'_t$, respectively. Increasing $\phi^\prime$ thus increases the autocorrelation (but not the variance) of the variable of interest $S_t$, while increasing $\phi^\prime$ increases the autocorrelation (but not the variance) of the spurious variable $U_t$.

The parameter $\gamma \geq 0$ determines the signal strength, or the strength of the dependence of $Y_t$ on $S_t$. Testing $H_0$ in Equation \ref{eqn:original_null} is equivalent to testing $H_0: \gamma=0$. Ideally, our method should also be able to identify local regions in the sample space $\mathcal{S}$ where the two distributions are different or the same. To assess our method's local performance, we hence include a hard thresholding operator $H_\delta(\cdot)$ in Equation~\ref{eqn:synthetic}, which, regardless of the signal strength $\gamma$, creates a region in $\mathcal{S}=\mathbb{R}$ where $\mathbb{P} (Y=1 | S=s) = \mathbb{P} (Y=1)$ for $s \in (-\delta, \delta)$.

The two parameters $\phi$ and $\phi^\prime$ allow us to create synthetic data with the dependence structures in Figure \ref{fig:settings}. More specifically,\\
{\bf Setting A:} $\phi^\prime=\phi=0$\\
\textbf{Setting B:}
$\phi^\prime>0$ induces autocorrelation in $\series{S_t}$ via $U'_t$; $\phi=0$\\
\textbf{Setting C:} $\phi^\prime>0$ induces autocorrelation in $\series{S_t}$ via $U^\prime_t$, while $\phi>0$ induces autocorrelation in $\series{Y_t}$ via $U_t$.
\vspace{.5em}

\begin{remark} In this toy example, $S_t=U'_t$ with $U'_t$ observed. More generally, $\S_{<t}$ can depend on unmeasured variables $\mathbf{U}'_t$, as well as confounding variables $\Z_t$, which are connected to both $\S_{<t}$ and $Y_t$. The latter setting is discussed in Section \ref{sec:confounding}. We also note that the spurious variables $\mathbf{U}_t$ and $\mathbf{U}'_t$ can have more complex temporal dependence (than in the example), as indicated by the fully connected sequences in Figure \ref{fig:setting3}. 
\end{remark}

\subsection{Test Results for Synthetic Data}
The logistic generative model in Equation~\ref{eqn:synthetic} provides a variety of controls over the dependence structure of $\series{(S_t,Y_t)}$. For our synthetic experiments, we implement Algorithm~\ref{alg:distributional-differences} with either the MC bootstrap or the permutation test for $\lvert\mathcal{V}\rvert=250$ evaluation points. We estimate the regression function $m_\text{post}(s)= \mathbb{P} (Y_t=1 | {S}_t={s})$ using train data $\{(S_t, Y_t)\}_{t \in \mathcal{T}_1}$ and a Nadaraya-Watson (NW) kernel estimator with an Epanechnikov kernel and the bandwidth chosen as the sample standard deviation of $s$ divided by $\lvert\mathcal{T}_1\rvert^{1/5}$ \citep{li2007nonparametric}. For the bootstrap test, we estimate the label distribution $m_{\text{seq}}$ using an order $k=4$ Markov chain and labels $\{Y_t\}_{t \in \mathcal{T}_2}$.

\textbf{Validity:} Testing $H_0$ in Equation \ref{eqn:original_null} is equivalent to testing $H_0: \gamma=0$ (no signal strength). To examine validity, we set $\gamma=0$ and simulate 500 independent data sets for each experiment, or combination of ``setting'' (A, B, C) and ``test method'' (permutation or MC bootstrap test with $\lvert\mathcal{T}_1\rvert=\lvert\mathcal{T}_2\rvert=\lvert\mathcal{V}\rvert=250$). That is, each experiment returns 500 p-values. If the test controls type I error, we expect these p-values to be approximately uniformly distributed. Figure \ref{fig:setting-C-QQ} assesses validity by plotting the difference between the empirical and (uniform) theoretical quantiles against the theoretical quantiles; this is {equivalent to} a standard quantile-quantile plot with the diagonal subtracted. As a baseline, we provide a 95\% confidence interval of this difference based on $10,000$ Monte Carlo simulations of 500 uniformly distributed random variables.

The permutation test (left panel) is valid under Settings A and B, where $Y_t$ and $Y_{t-1}$ are independent after conditioning on $S_t$. However, under Setting C, the p-values tend to have lower values than a uniform distribution, corresponding to higher-than-nominal type I errors at most significance levels. The MC bootstrap test (right panel) controls the type I error at all significance levels for all three dependence settings, indicating that our adjustment to account for the dependence in $Y_t\mid S_t$ achieved the desired result.

\begin{figure}[!tbh]
    \centering
    \includegraphics[width=\linewidth]{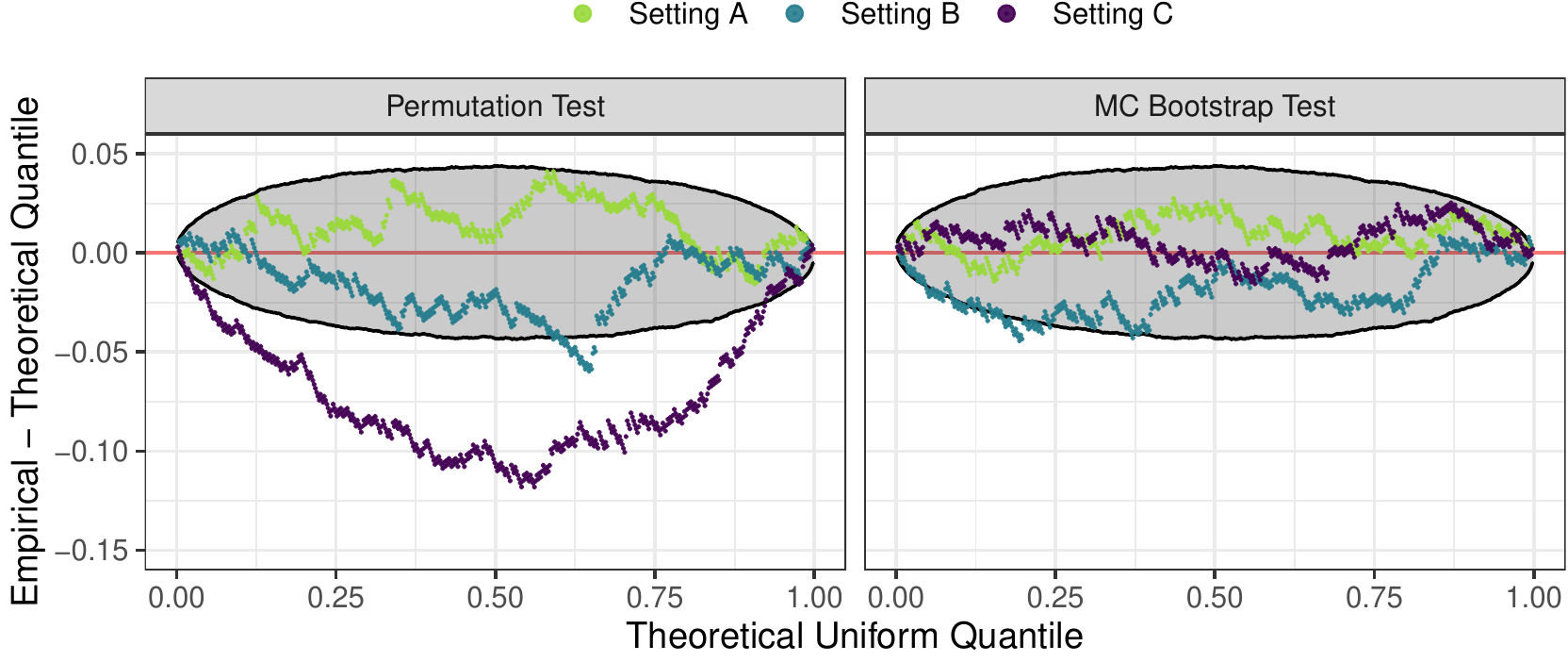}
    
    \caption{\small{{\bf Synthetic sequence data: Validity of permutation test ({\em left}) versus Markov Chain bootstrap test ({\em right}).} Under $H_0: \gamma=0$, a valid test is expected to return uniformly distributed p-values. Each curve corresponds to a different experimental setting (A, B, or C), and shows the difference between empirical and uniform theoretical quantiles for 500 repetitions; see text for details. The gray region represents a 95\% pointwise confidence interval derived from Monte Carlo samples of 500 uniform deviates. Under Setting C (labels dependent even after conditioning on predictors; purple curve), the permutation test (left panel) does not control the type I error, but the Markov Chain bootstrap test (right panel) does. (Setting A: $\phi=\phi^\prime=0$. Setting B: $\phi=0, \ \phi^\prime=0.8$. Setting C: $\phi=\phi^\prime=0.8$. Sample sizes $\lvert\mathcal{T}_1\rvert=\lvert\mathcal{T}_2\rvert=\lvert\mathcal{V}\rvert=250$.) }}
    \label{fig:setting-C-QQ}
\end{figure}

\textbf{Power:} We next examine how the power of the test $H_0:\gamma=0$ versus $H_1:\gamma \neq 0$ depends on
\begin{enumerate}[label=\roman*.]
  \item  signal strength  $\gamma$ (Figure \ref{fig:setting-C-power}, top), 
  \item train sample size $\lvert\mathcal{T}_1\rvert$ (Figure \ref{fig:setting-C-power}, bottom),
  \item autocorrelation $\phi$ in labels $\series{Y_t}$, or equivalently, correlation in $\series{Y_t \rvert S_t}$ (Figure \ref{fig:power-versus-dependence}, left), and 
  \item autocorrelation $\phi^\prime$ in predictors $\series{S_t}$ (Figure \ref{fig:power-versus-dependence}, right).
\end{enumerate}
At each fixed value of $\gamma$, we perform 1000 simulations. Power is then estimated as the fraction of rejected null hypotheses at the $\alpha=0.05$ level. To ensure validity, we choose $\lvert\mathcal{T}_{1}\rvert\ge250$ as before for the MC bootstrap test. (The permutation test is valid by construction under Settings A and B, but not Setting C.)

\begin{figure}[!p]
    \centering
    \includegraphics[width=.8\linewidth]{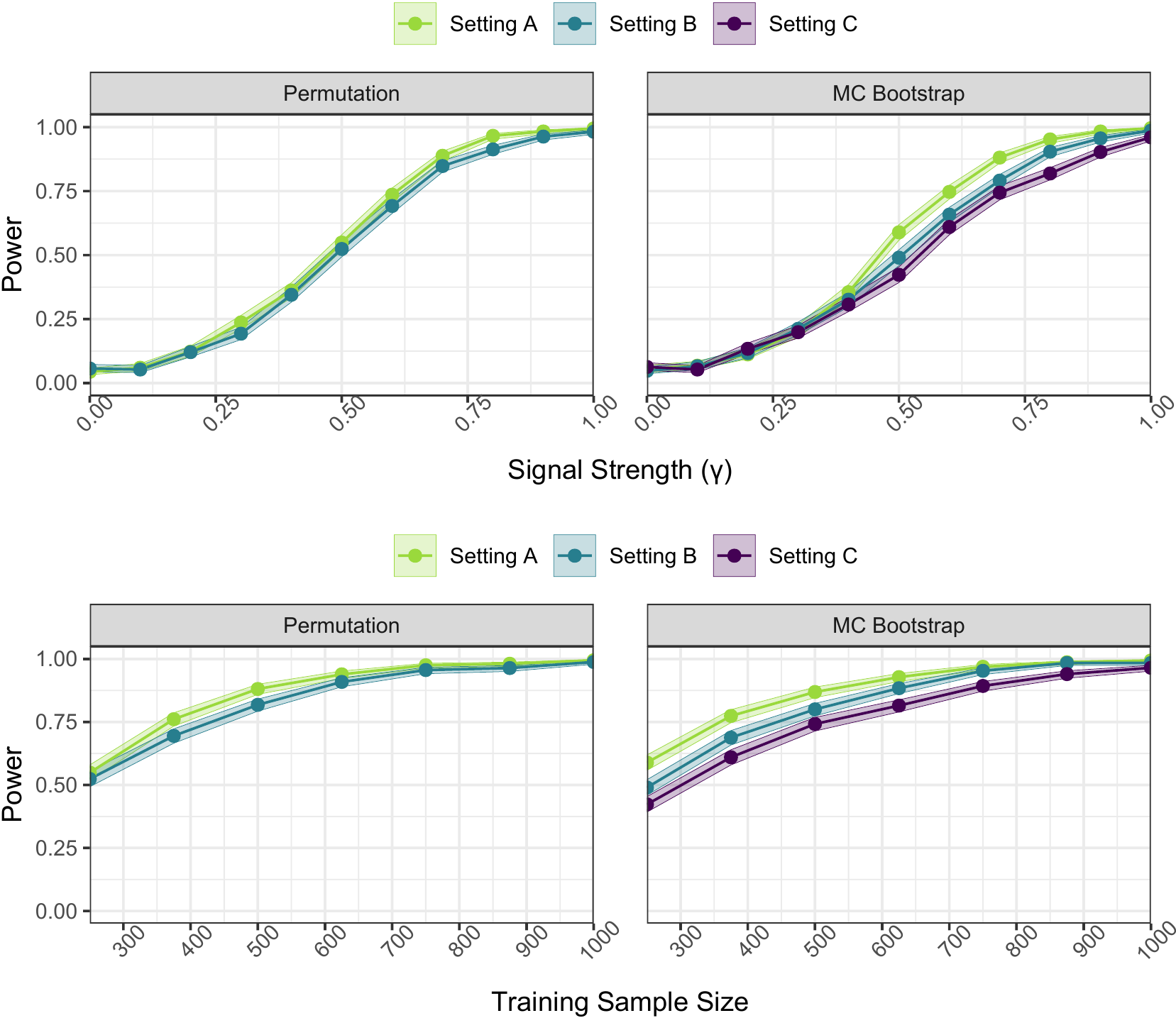}
    \caption{
    \small{\textbf{Synthetic sequence data: Power as a function of signal strength and training sample size.}  
    {\em Top:} The power of all tests increases with the signal strength $\gamma$, regardless of dependence setting. The MC bootstrap test has similar power as the permutation test, but the former test can be applied to the more challenging Setting C with dependent labels $Y_t \vert S_t$. Sample sizes $|\mathcal{T}_1|=|\mathcal{T}_2|=|\mathcal{V}|=250$. {\em Bottom:} The power of all tests, at the alternative $\gamma=0.5$, increases with the train sample size $|\mathcal{T}_1|$, regardless of dependence setting. Sample sizes $|\mathcal{T}_2|=|\mathcal{V}|=250$. The filled regions represent 95\% pointwise confidence intervals for binomial  proportions. (Setting A: $\phi=\phi^\prime=0$. Setting B: $\phi=0, \ \phi^\prime=0.8$. Setting C: $\phi=\phi^\prime=0.8$.) Setting C is not shown for the permutation test, as it is not valid.}}
    \label{fig:setting-C-power}
\end{figure}

\begin{figure}[!p]
    \centering
    \includegraphics[width=.8\linewidth]{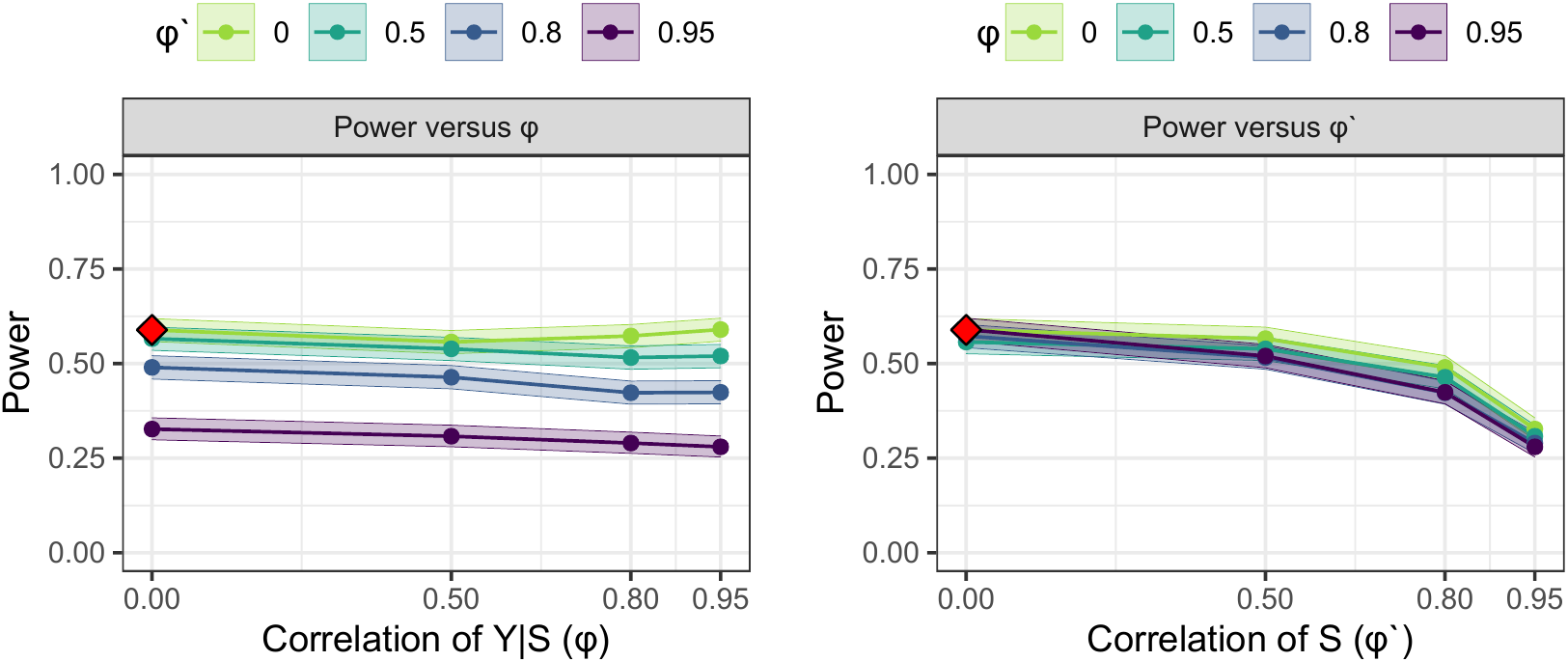}
    
    \caption{\small{{\bf Synthetic sequence data: Power of bootstrap test as function of autocorrelation of $Y$ ({\em left}) and $S$ ({\em right})}. The red diamond-shaped markers correspond to Setting A with \texttt{IID} data. The power to detect the signal $\gamma=0.5$ is independent of the correlation in $Y_t|S_t$ (value of $\phi$), but decreases with correlation in $S_t$ (larger values of $\phi^\prime$). The filled regions represent 95\% pointwise confidence intervals for binomial proportions. (Sample sizes $|\mathcal{T}_1|=|\mathcal{T}_2|=|\mathcal{V}|=250$.)}
    \label{fig:power-versus-dependence}}
\end{figure}

As expected, the power increases with the signal strength $\gamma$ for all tests and dependence settings (Figure~\ref{fig:setting-C-power}, top). When both tests are valid, the MC test has the same power as the permutation test. The practical implication is that, even if one \emph{thinks} Setting B is a good approximation to the problem at hand, there are benefits to applying the MC bootstrap test: one can achieve similar power with the advantage of having robustness in the event that the labels are dependent after conditioning on predictors.

Figure~\ref{fig:setting-C-power} (bottom) indicates that the power of the tests may be determined by the {\em quality of the regression estimator} $\widehat{m}_\text{post}(\s)$: indeed, Figure~\ref{fig:setting-C-power} (bottom) shows that the power at a fixed alternative ($\gamma=0.5$) increases with the train sample size $\lvert\mathcal{T}_1 \lvert$. The latter result is consistent with theorem 3.3 of \cite{kim2019global}, which states for a regression permutation test under Setting A, that if the chosen regression method $\widehat{m}_\text{post}(\s)$ has a small mean integrated squared error, then the power of testing (\ref{eqn:new_null}) is large over a wide region of alternative hypotheses.

\begin{figure}[!t]
    \centering
    \includegraphics[width=\linewidth]{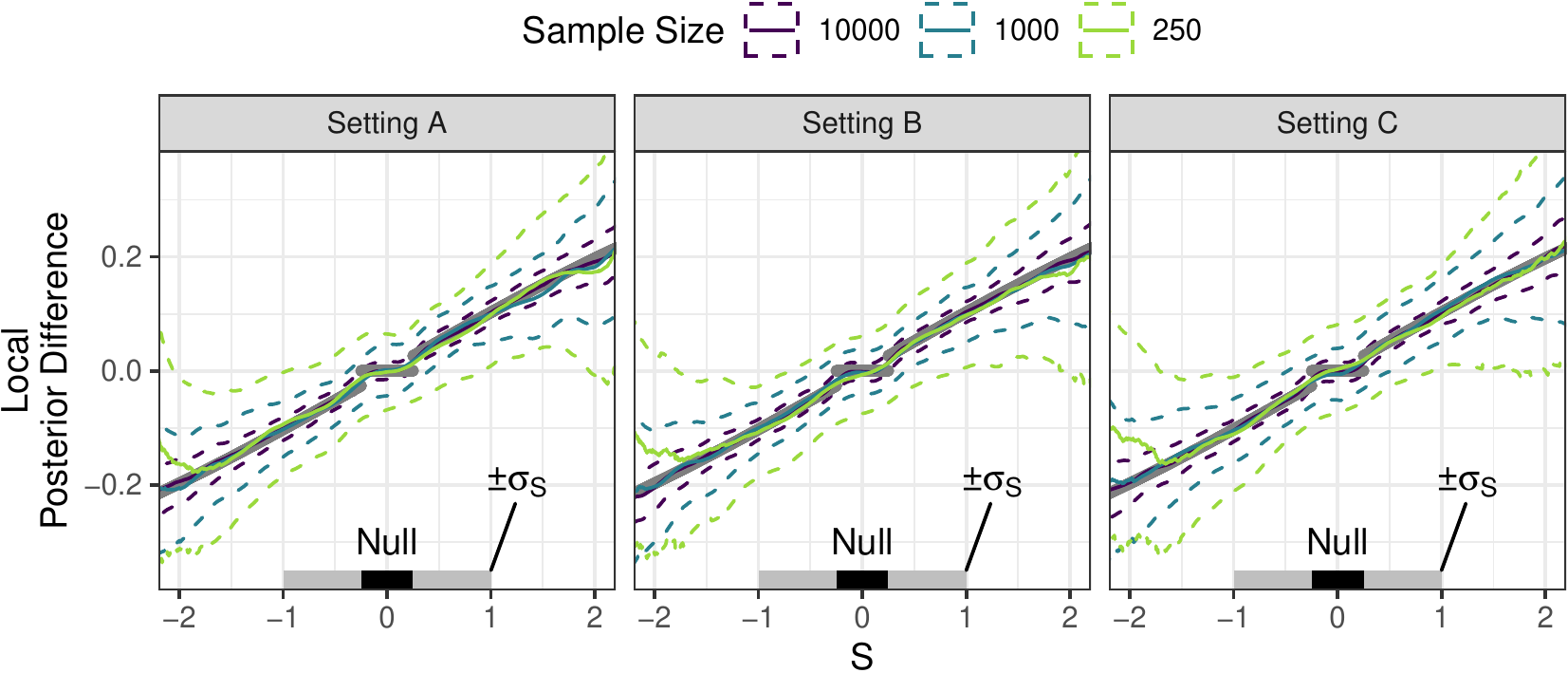}
    
    \caption{\small{{\bf Synthetic sequence data: Local posterior difference in null region.} The gray curves in each panel/setting show the true LPD, which is zero in the local null (no signal) region $s\in (-0.25,0.25)$. The solid and dashed colored curves represent the mean and one standard deviation estimates of the LPD over 200 simulated data sets. These estimates are on average close to the true LPD, with the dispersion decreasing for a local nonparametric estimator (like the NW kernel estimator) as the train sample size increases. The gray bar at the bottom marks the standard deviation of $s$; the large variance of the estimated LPDs far from $s=0$ is partially due to the concentration of data near $s=0$. (Setting A: $\phi=\phi^\prime=0$. Setting B: $\phi=0, \ \phi^\prime=0.8$. Setting C: $\phi=\phi^\prime=0.8$. Sample sizes $\lvert\mathcal{T}_1\rvert$ vary, $\lvert\mathcal{T}_2\rvert=\lvert\mathcal{V}\rvert=250$.)}}
    \label{fig:setting-C-importance}
\end{figure}

Figure \ref{fig:power-versus-dependence} brings insight on how dependence in $\series{(S_t,Y_t)}$ affect the power of the MC bootstrap test. The red diamond-shaped markers represent Setting A with \texttt{IID} sequence data ($\phi=\phi^\prime=0$). Increasing  correlation in the labels $Y_t \vert S_t$ (larger values of $\phi$) has no effect on power, while the test remains valid as long as $\widehat{m}_\text{seq}$ is accurate; this further emphasizes the previous result that the bootstrap test is robust to correlation in $Y_t\vert S_t$ without sacrificing power. Meanwhile, increasing correlation in $S_t$ (larger values of $\phi^\prime$) reduces power; this follows from a reduced effective sample size which in turn reduces the quality of $\widehat{m}_\text{post}$.

\textbf{Local Posterior Differences:} 
For our synthetic example, the hard thresholding operator $H_\delta(\cdot)$ induces a region $s\in (-\delta,+\delta)$ where $p (s | Y=1) = p (s | Y=0)$. If the null hypothesis in Equation \ref{eqn:new_null} is rejected, then the estimated LPDs  can identify the regions of large versus small distributional differences, as long as the regression estimator $\widehat{m}_\text{post}$ is consistent and the train sample size $|T_1|$ is sufficiently large. Figure~\ref{fig:setting-C-importance} shows the average and one standard deviation estimates of the LPD for a NW kernel estimator over 200 simulations.

\section{Relating Evolution of TC Convection to Rapid Intensity Change}\label{sec:results_TC}
In our TC study, each observation consists of (i) a 24-hour sequence $\S_{<t}=\{\X_t,\X_{t-1},\dots,\X_{t-48}\}$ of one-dimensional radial profile functions $\X_t=\frac{1}{2\pi}\int_0^{2\pi}T_{b,t}(r,\theta)d\theta$, sampled every 30 minutes for a total of 48 profiles (see Figure \ref{fig:profiles}), and (ii) a binary label $Y_t \in \{0,1\}$ for the entire sequence\footnote{{We check that the dominant principal components of $\X_t$ and the continuous intensities used to derive $Y_t$ are stationary via augmented Dickey-Fuller tests; the p-value of each test (including tests of the first 3 ORB coefficients in \citealt{mcneely2020unlocking}) are all $<10^{-20}$. We conclude that Assumption \ref{assump:stationary} is reasonable for these data.}}. Since all individual sample points that are part of a rapid intensity change event are labeled as $Y=1$, a 24-hour sequence $\S_{<t}$ with sequence label $Y_t=1$ could either be part of an ongoing RI/RW event (if $Y_t=1$ falls near the end of an event), or be part of the lead-up to RI/RW (if $Y_t=1$ falls near the beginning of an event).  Analyses of the latter case --- such as approaches that can identify archetypal modes of structural evolution preceding the {\em onset} of RI/RW --- are particularly valuable to forecasting of RI/RW events.

We divide our TC study into three parts:
\begin{enumerate}
    \item  Analysis by event type (RI versus not RI, or RW versus not RW) within each basin, North Atlantic (NAL) or eastern North Pacific (ENP), for a total of four different two-sample tests.
    \item Case studies of three tropical cyclones (Hurricanes Nicole, Jose, and Rosa).
    \item Analysis of a subset of our data that consists of sequences immediately preceding RI onset. (We ask whether our regression two-sample test can find archetypical evolutionary modes preceding RI onset, and if so, what the ``lead time'' between typical patterns and the RI onset would be.)
\end{enumerate}
As in the synthetic example, we estimate $m_\text{seq}$ in Algorithm~\ref{alg:distributional-differences} with a Markov chain of order $k=8$. The sequence data $\S_{<t}$ are however much more complex than in our synthetic example. This is where we benefit from a more complex regression method for estimating $m_\text{post}$; here we fit a convolutional neural network to the 24-hour sequence data. Further details on how we estimate $m_\text{seq}$ and $m_\text{post}$ can be found in \cite{mcneely2022Supplement} (Section A).

\subsection{Analysis by Event Type and Basin} 
\begin{figure}[!htbp]
    \centering
    \includegraphics[width=.9\linewidth]{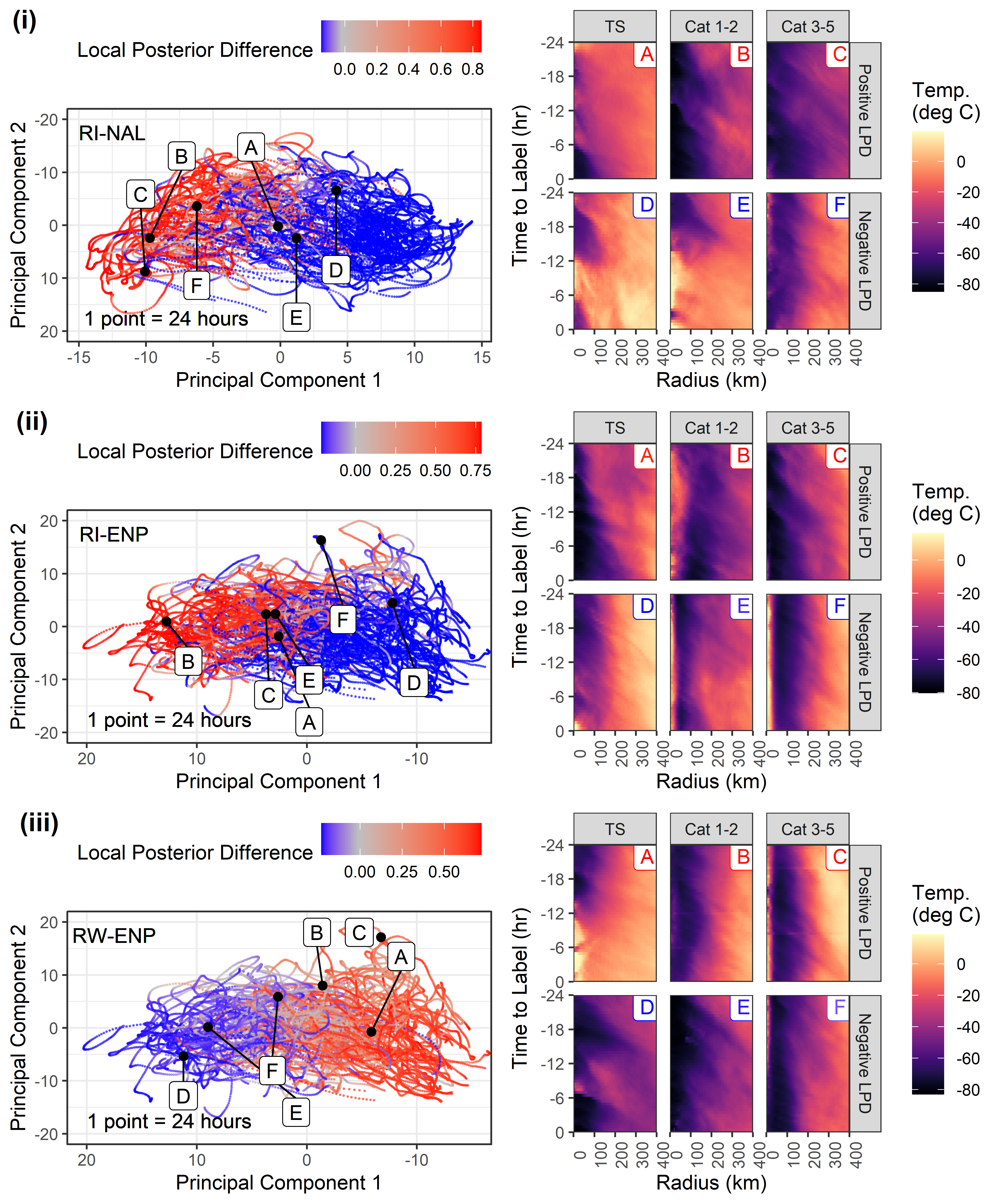}
    
    \caption{\small {\bf Analysis by event type and basin.} The MC bootstrap test rejects $H_0$ (that is, it detects significant differences in convection) for the RI-NAL, RI-ENP, and RW-ENP models; the test is not rejected for the RW-NAL model. Analyses of local posterior difference (LPD) are shown for the first three models (see panels i-iii). {\em Left column}: Two-dimensional PCA map of sequence data. One point in the map represents a 24-hour structural trajectory (sequence of radial profiles) $\S_{<t}$ with the color coding for the estimated LPD. {\em Right  column}: Six 24-hour structural trajectories at locations \boxit{A}-\boxit{F} in the PCA map, shown as Hovm\"oller diagrams (recall Figure~\ref{fig:profiles}). The examples for each study are selected at random from each combination of LPD sign (positive LPD: \boxit{A}-\boxit{C}; negative LPD: \boxit{D}-\boxit{F}) and TC intensity (Tropical Storm: \boxit{A}, \boxit{D}; Category 1-2: \boxit{B}, \boxit{E}; Category 3-5: \boxit{C}, \boxit{F}). See text for a discussion of the results.}
    \label{fig:basinplot}
\end{figure}

{\em Significance test.} We start by testing $H_0:{p(\s_t|Y_t=1)}={p(\s_t|Y_t=0)}$ by event type and basin. For rapid intensification (RI), the MC bootstrap test rejects $H_0$ at level $0.05$ for both the NAL and ENP basins, meaning that we indeed detect a significant difference (p$<$0.01) in 24-hour sequences $\S_{<t}$ of convective structure leading up to RI versus not-RI events. For rapid weakening (RW), the MC bootstrap test rejects $H_0$ in the ENP basin, but not in the NAL basin.

Our results are consistent with scientists' understanding of TCs; for rapid intensification, a TC exhibits a narrow range of convective patterns ``primed'' to efficiently convert heat energy to mechanical energy across the storm, hence the structural difference in convection for RI versus not-RI events. Rapid weakening, on the other hand, is a more complex process driven by several factors external to the TC, such as vertical wind shear, which may not be fully captured by convective structure. In addition, RW is expected to be more difficult to detect in the NAL basin due to the broader range of possible environmental configurations and the increased rarity of over-water RW in the basin.

{\em Local posterior difference.} Next we investigate what kind of structural patterns lead to the rejection of $H_0$ for the RI-NAL, RI-ENP, and RW-ENP models. Figure \ref{fig:basinplot} (left) shows a two-dimensional embedding of the sequence data via principal component analysis (PCA; computed separately for each basin). Each point represents a 24-hour sequence $\S_{<t}$ colored by its local posterior different (LPD). Note that PCA is only used for purposes of visualization; the test itself is performed on the entire sequence of radial profiles without a prior dimension reduction step. Figure \ref{fig:basinplot} (right) shows examples of Hovm\"oller diagrams for six 24-hour sequences $\S_{<t}$ sorted by LPD and TC intensity.

TCs are known to have different distributions of $\S_{<t}$ for different basins. Nevertheless, we identify the same type of evolutionary patterns of convective structure for RI-NAL (panel i) and RI-ENP (panel ii): Positive LPD or ``high chance of RI'' (see diagrams \boxit{A}-\boxit{C} for i and ii) tends to occur for cold cloud tops near the core (dark blue at smaller radius), growing in coverage and depth of convection with time (dark blue region extending to larger radii when going from -24 to 0 hours). Meanwhile, negative LPD or ``low chance of RI'' (see diagrams \boxit{D}-\boxit{F} for i and ii) tends to occur when TCs already possess a well-defined eye (narrow yellow region near the center) or exhibit decaying core convection (dark blue region decreasing in size when going from -24 to 0 hours). {While such patterns can be directly quantified and studied in future works, this work remains focused on exploration of entire radial profiles.}

Finally, RW-ENP results (panel iii) are not exactly opposite of the RI-ENP results (panel ii), meaning that ``high chance of RW'' patterns might not mirror ``low chance of RI'' patterns, and vice versa. In particular, TCs commonly form strong convective cores without eyes (iii-\boxit{E}) \textit{prior} to intensification, form an eye (end of i-\boxit{C}) during RI, then rapidly weaken by dissipating entirely (ii-\boxit{D}, iii-\boxit{A}, iii-\boxit{B}) without reforming a cold, eyeless core because the reduction of intensity is accompanied by a collapse of convection throughout the TC. Thus, RI/RW-ENP results are not symmetric. Unfortunately, the RW-ENP model also predominantly captures the trivial result that currently-intense TCs are more likely to weaken; see Section \ref{sec:confounding} for a discussion of potential corrections.

\subsection{Hurricane Case Studies}
\begin{figure}[!htbp]
    \centering
    \includegraphics[width=.8\linewidth]{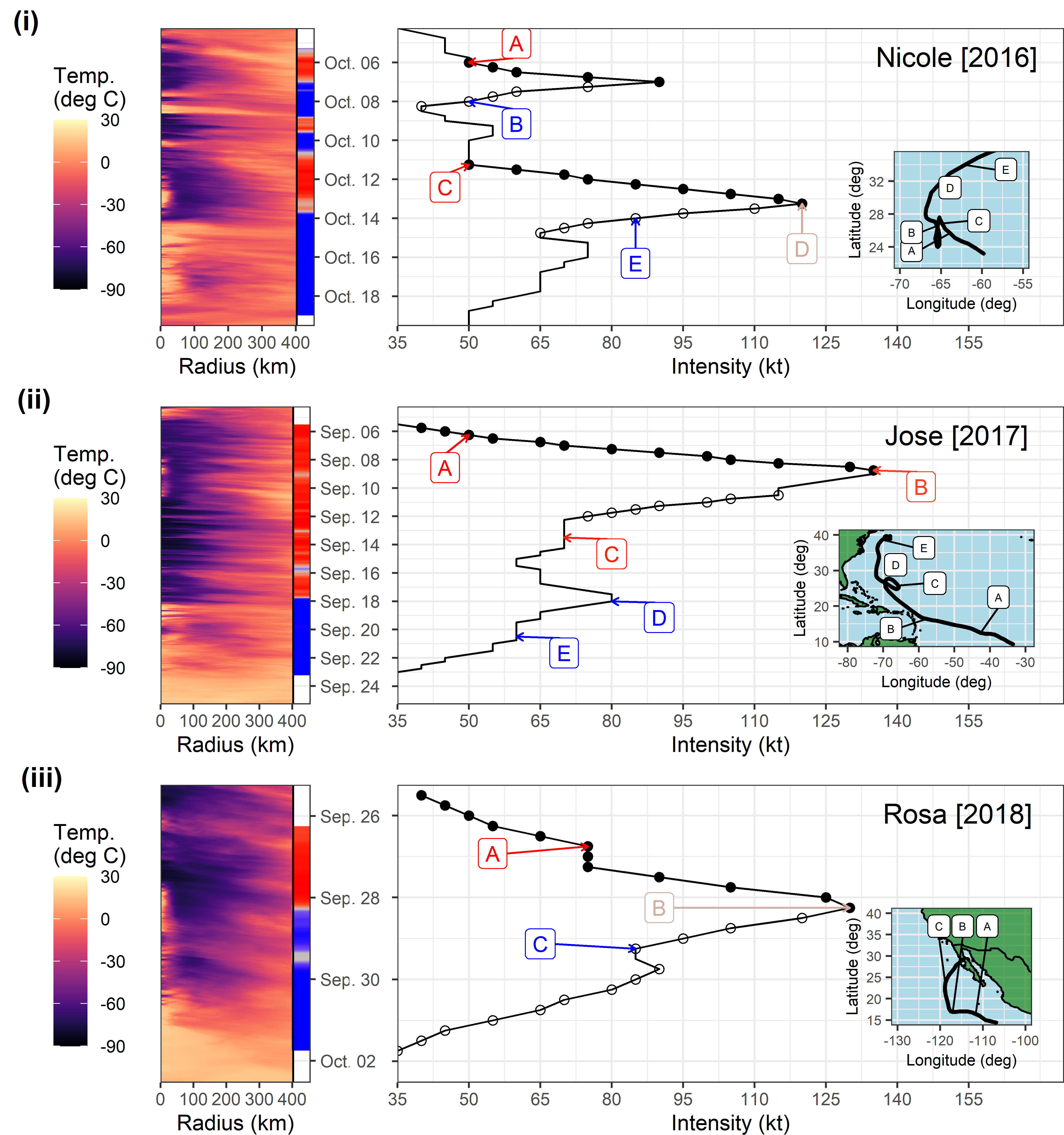}
    
    \caption{\small {\bf
    Case studies of Hurricanes Nicole {[NAL~2016,~i]}, Jose {[NAL~2017,~ii]}, and Rosa {[ENP~2018,~iii]}.} {\em Left}: Structural trajectories of each storm through its entire lifetime, with the sidebar to the right showing the LPD at each 30-minute observation; each LPD value is evaluated using the preceding 24-hour sequence for an RI model (trained on the train sample for that basin), with positive values in red and negative values in blue. {\em Right}: Storm intensity over time, where filled versus empty circles mark RI and RW events, respectively. The physical track of the storm across the North Atlantic is shown inset, with the same time points labeled.\\
    \emph{Panel (i)}: For Hurricane Nicole, the LPDs capture rapidly intensifying periods (\boxit{A},~\boxit{C},~\boxit{D}), collapsing convection \boxit{B}, and the decay of the TC eye \boxit{E}.\\
    \emph{Panel (ii)}: Hurricane Jose was subjected to high vertical wind shear near Sept. 9, which our model does not account for; while the core convection of the TC remained poised for RI, high shear instead caused rapid weakening. \\
    \emph{Panel (iii)}: Hurricane Rosa exhibited two interesting phenomena captured by the LPDs. First, it experienced a pause in its rapid intensification \boxit{A}; the LPDs indicate an ongoing RI threat at this time, consistent with the resumption of intensification 18 hours later. Second, beginning on Sept. 28, the TC underwent an eyewall replacement cycle, associated with weakening. The LPDs mark this shift at the TC's peak intensity, as well as the brief period (following \boxit{C}) where eyewall replacement is completed prior to landfall.}
    \label{fig:case_studies}
\end{figure}

Thus far, we have analyzed the collection of 24-hour sequences in the 2013-2020 test sample as a whole; however, forecasters monitor {\em individual} storms in real time for signals of RI. Here, we take an in-depth look at our results for three individual TCs in the test set: Hurricanes Nicole, Jose, and Rosa. Each of these storms display distinct evolutionary modes. We track each TC through its lifetime to investigate the relationship between the evolution of convective structure, the LPDs, and intensity change.

Figure \ref{fig:case_studies}, left, shows the evolution of radial profiles for each storm. Appended to the right of these profiles is a sidebar showing the LPDs from the RI model for the associated basin. The right panels display the evolution of each storm's intensity, where filled and hollow markers indicate RI and RW events, respectively. The storm's physical track over the ocean is shown as an inset. These three hurricane case studies respectively highlight: (i) signals which lead RI, (ii) the effect of vertical wind shear, and (iii) the appearance of an eyewall replacement cycle (a process by which a second eyewall forms, robbing the TC of energy as it shrinks to replace the original eyewall).

Panel (i) depicts {\em Hurricane Nicole} [2016]. The TC underwent RI \boxit{A} on Oct. 6 before its intensity stalled while influenced by the outflow from Hurricane Matthew [2016], which induced vertical wind shear \boxit{B}. Several days later, Nicole re-intensified \boxit{C}, \boxit{D} before again weakening and then transitioning into an extratropical system \boxit{E}. The LPDs in Figure \ref{fig:case_studies} appear to align with intensity changes and the emergence of deep convection \boxit{A}, \boxit{C} and eye formation \boxit{D}. The structural trajectories immediately preceding \boxit{A} and \boxit{C} are particularly interesting: the TC has not yet begun to intensify rapidly, but 24-h sequences including and prior to \boxit{A} and \boxit{C} have strongly positive LPDs. These results indicate that structural trajectories of radial profiles in the North Atlantic may contain signals of RI prior to onset.

Panel (ii) shows {\em Hurricane Jose} [2017] and highlights the importance of vertical wind shear. This TC exhibited deep convection in the core for nearly two weeks, remaining at elevated RI risk according to the RI-NAL posterior differences. However, after an initial period of RI, high vertical wind shear disrupted the TC structure around Sept. 9. The TC decayed from 135 kt to about 70 kt and never appreciably intensified again despite several periods of elevated LPDs prior to \boxit{D}. The underlying regression $\widehat{m}_\text{post}$ does not account for the vertical wind shear which prevented the TC's intensification. See Section \ref{sec:confounding} for a discussion about how to potentially account for external factors such as vertical wind shear.

Finally, panel (iii) depicts the short-lived eastern North Pacific TC {\em Hurricane Rosa} [2018], which underwent an extended period of RI before beginning an eyewall replacement cycle. This evolution included two interesting phenomena. First, the TC experienced a pause in its rapid intensification \boxit{A}; the LPDs indicate an ongoing RI threat at this time, consistent with the resumption of intensification 12-18 hours later. Second, beginning on Sept. 28, the TC underwent an eyewall replacement cycle, which manifests as a expansion of the eye accompanied by an evening out of convection across the storm. Such cycles typically result in a decrease of the TC's maximum sustained winds. The LPDs mark this shift at the TC's peak intensity, as well as the brief period \boxit{C} where eyewall replacement is completed prior to landfall.

\subsection{Structural Trajectories Preceding Rapid Intensification}
\begin{figure}[!tb]
    \centering
    \includegraphics[width=\linewidth]{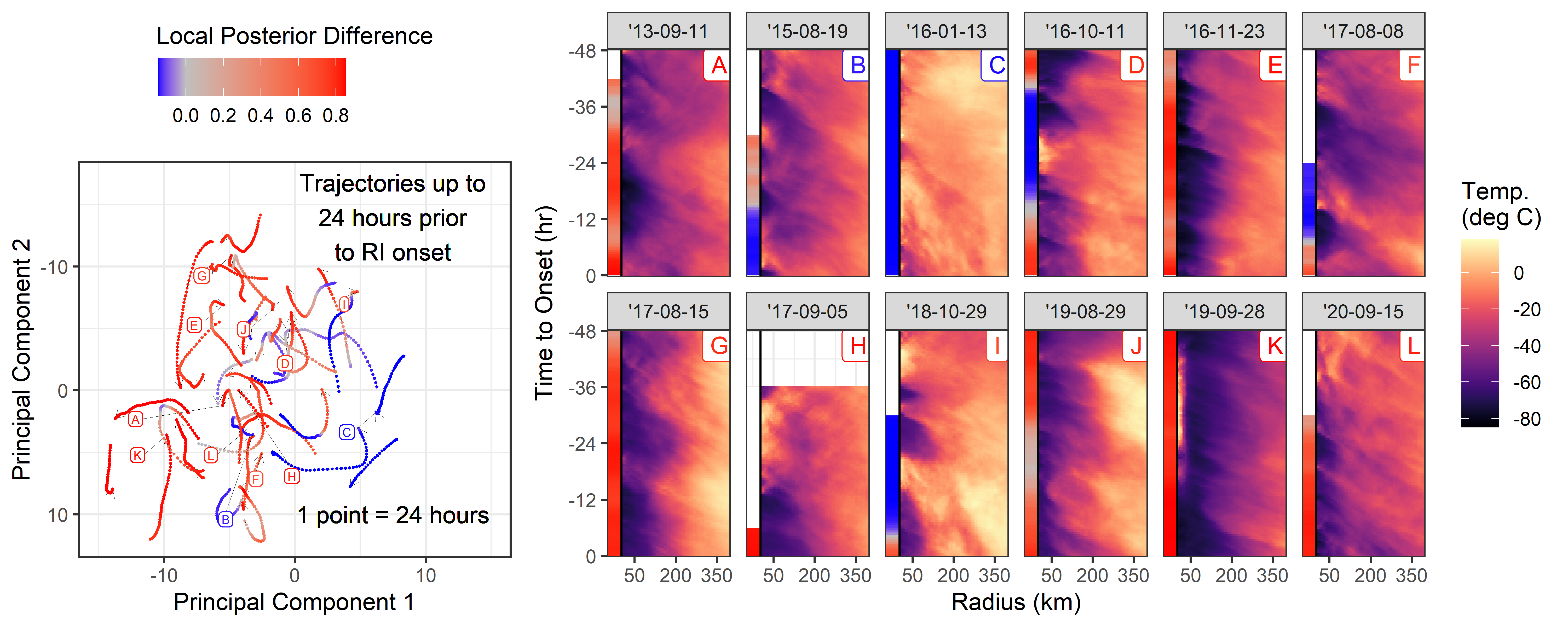}
    
    \caption{\small {\bf Detecting RI prior to onset.} \emph{Left}: PCA map of a subset of the RI-NAL test sample, representing 24-hour sequences prior to RI onset. We show 36 events with at least 24 hours of non-missing radial profiles $X_t$ (that is, at least one sequence $\S_{<t}$). \emph{Right}: Examples of 48-hour structural trajectories preceding RI onset. The local posterior differences, shown in the bar to the left of each map, largely signal TCs with deep core convection in the last 24 hours as ``primed for RI''.}
    \label{fig:ri-study}
\end{figure}

In the previous section, case studies indicated that positive LPDs can  \textit{lead} RI events. That is, they may signal convective structure primed for RI \textit{prior} to RI onset. That core convection can predict RI is known; at least one RII forecast model used by the NHC includes the fraction of GOES pixels between 50 and 200 km with temperatures below $-30^\circ$C as a predictor \citep{kaplan2015rii}. We would thus expect our LPDs to contain some signal leading RI.

In this section, we hone in on a subset of the RI-NAL test sample, which represents complete 24-hour sequences $\S_{<t}$ leading up to RI onset points (such as points i-\boxit{A} and i-\boxit{C} in Figure \ref{fig:case_studies}). This leaves us with a total of 36 sequences, visualized in the PCA map of Figure \ref{fig:ri-study}, left. We then investigate whether we are able to detect RI within the 48 hours prior to onset, and if so, what the lead time might be. As before, a positive posterior difference indicates $p(\s_t|Y_t=1)>p(\s_t|Y_t=0)$. Because $\S_{<t}$ here encodes the recent evolution of convective structure, we interpret LPD $\lambda(\s_t)\gg0$ prior to RI  as an indication that the TC at time $t$ exhibits convective structures primed for RI onset. That is, $\lambda(\s_t)\gg0$ in these cases indicates an {\em elevated risk of RI} based on convective evolution alone. 

In 30 (83\%) out of 36 onsets, the RI-NAL LPDs indicate an above-average RI threat at onset. In 28 (78\%) and 24 (66\%) onsets, LPDs indicate an above-average RI threat 6 or 12 hours prior to onset, respectively. Some example trajectories are shown in Figure \ref{fig:ri-study}, right. Prior to RI onset, TCs tend to exhibit deep core convection, evidence of strong outflow (indicated by the downward slant of $\S_{<t}$), and often pronounced diurnal cycles (oscillations in $\S_{<t}$ over time). These results indicate that structural trajectories or LPDs from our test could serve as valuable inputs to RI forecast models such as the SHIPS-RII \citep{kaplan2015rii}.

\section{{Limitations and Potential Extensions of Methods}}
\label{sec:extensions}
{We end our paper by discussing limitations of the current analysis, and suggest potential extensions of our work  to answer different questions about the behavior of TCs in the same binary (RI versus not-RI) setting.}

\subsection{Adding Multiple Data Sources, {Environmental Variables, and Other Functional Features}}\label{sec:discussion_multivar}
Deep convection as revealed by IR imagery is only one ingredient required for RI onset. {Our current analysis does not include environmental variables, such as vertical wind shear, which are known to affect RI. Our analysis also neglects asymmetric patterns in radial profiles, and key ORB functional features. These limitations however only pertain to the presented TC analysis, rather than being limitations of the proposed methodology per se. Indeed,
our testing framework for distributional differences for a binary response variable extend to settings with multiple data sources and different data types:}

For a joint analysis of several variables, one directly adds predictors to the regression model $\widehat{m}_\text{post}$ so that the sequence $\S_{<t}$ consists of, for example, both structural features (e.g., ORB functions) and environmental features (e.g., vertical wind shear or sea surface temperatures). Furthermore, the CNN model can easily handle multiple inputs; including {other observation bands (e.g., the water vapor-sensitive $6.9\mu$m GOES band or microwave imagery from polar-orbiting satellites),} and ORB functions  other than radial profiles {(e.g., the size of temperature level sets as a function of the temperature threshold; see \citealt{mcneelyquantifying,mcneely2020unlocking}). These inputs are then included as channels, which are combined at the dense layer} (with separate feature extraction layers), or added together at the logistic layer (a neural additive model with separate feature and dense layers for each input).

{In particular, the current TC analysis of GOES IR-imagery can be extended to capture asymmetry in convective structure by separating the radial profiles into four quadrants. Our ongoing work on structural forecasting illustrates the promise of such an approach \citep{mcneely2022structural}.}

\subsection{Adjusting for Other Variables}\label{sec:confounding}

{Our presented framework detects potential associations between high-dimensional time series and binary labels; any detected association might of course be due to variables that confound the relationship between the label and the covariates of interest. Our general methodology can however be extended to account for confounding variables with the test of independence then generalizing to a test of conditional independence, and the corresponding regression problem including additional covariates:}

Suppose that we want to detect distributional differences in sequence data  $\{\S_{<t}\}_{t\ge0}$  preceding an event $Y_t=1$ versus a non-event $Y_t=0$, after adjusting for the effect of other variables with sequence data $\{\Z_{<t}\}_{t\ge0}$. For example, wind shear might confound the relationship between convective structures and RI or RW.
Assuming a stationary process $\{(\S_{<t}, \Z_{<t}, Y_t)\}_{t\ge0}$, hence omitting $t$, we test conditional independence
\begin{align}
   H_0: p (\s | Y=1, \Z=\z) &= p (\s | Y=0, \Z=\z),\ \text{for all} \ \s \in \mathcal{S} \text{ and all } \z,  
   \ \text{versus}   \label{eqn:original_null_CI} \\\ H_1: p (\s | Y=1, \Z=\z) &\neq p (\s | Y=0, \Z=\z), \ \text{for some} \ \s \in \mathcal{S} \text{ or } \z.\notag
\end{align}
These hypotheses are equivalent to
\begin{align} 
H_0: \mathbb{P} (Y=1 | \S=\s,\Z=\z) = \mathbb{P} (Y=1|\Z=\z),
 & \text{ for all} \ \s \in \mathcal{S} \text{ and all } \z,   \label{eqn:global_hypothesis_CI}\\
H_1: \mathbb{P} (Y=1 | \S=\s,\Z=\z) \neq \mathbb{P} (Y=1|\Z=\z),
&\text{ for some } \mathbf{s} \in \mathcal{S} \text{ or } \z.\notag
\end{align}
Analogous to Equation~\ref{eqn:test_stat}, we can define a regression test statistic $\lambda$ based on the difference between an estimate of $\mathbb{P} (Y=1 | \S=\s,\Z=\z)$ and $\mathbb{P} (Y=1|\Z=\z)$. Note that if $\Z$ is associated with both $\S$ and $Y$, then a permutation test is not valid even for Setting A with \IID\ data, because of lack of exchangeability under $H_0$. One solution for (\IID\ as well as \DID) sequence data $\{(\S_{<t}, \Z_{<t}, Y_t)\}_{t\ge0}$ is to extend our bootstrap test to a procedure where one estimates the distribution of the label series $\{Y_t\}_{t\ge0}$ {\em conditional on} $\Z$.

Regarding TC rapid intensity change, the admission of confounders would improve the interpretability of both the RI and RW tests. In the case of RI, wind shear $\Z$ is a powerful environmental predictor and can inhibit intensification $Y$ of a TC with otherwise favorable structure $\S$ (e.g., Hurricane Jose [2017], Figure \ref{fig:case_studies}, ii-C). Meanwhile, the results for RW (Figure \ref{fig:basinplot}(iii)) appear to weakly capture the obvious relationship: stronger storms are more likely to rapidly weaken. By accounting for the effect of current intensity ($\Z$), we could better assess the relationship between structural evolution ($\S$) and intensity change ($Y$).

\subsection{Local P-Values}
{In this work, we  refer to the local posterior difference (Equation~\ref{eqn:local-importance}) as a local diagnostic, rather than as a local p-value, because empirical results show that we do not control the type I error of a point-wise test  $H_0(\s): \mathbb{P} (Y=1 | \mathbf{S}=\mathbf{s}) = \mathbb{P} (Y=1)$ at current sample sizes. The LPD value can, however, in principle be used to test the local null \begin{align}
\label{eq:local_null_ball}
    H_0^{\epsilon}(\s): \mathbb{P} (Y=1 | \mathbf{S}=\mathbf{s}') = \mathbb{P} (Y=1)
\text{ for all } \mathbf{s}' \in B(\s;\epsilon),
\end{align}
where $B$ is a ball of radius $\epsilon$ centered at $\s$.  {In \citet{mcneely2022Supplement}, we} show that, under DAG B, the local p-values (\ref{eq:local_null_ball}) are valid if the regression estimator
$\widehat m_{\text{post}}$ only uses the observations in $\mathcal{D}$ such that $\S \in B(\s;\epsilon)$. The latter assumption holds for regression estimators that are based on partitions, such as tree-based estimators (random forests, boosting methods), as well as smoothing kernel estimators with finite support.}

\subsection{Bootstrapping the Label Series}\label{sec:boot}
To estimate the null distribution of the test statistic $\lambda$ (\ref{eqn:global_hypothesis}), we currently assume a $k$-step Markov chain and draw new labels from the Markov autoregressive model (\ref{eqn:markov-marginal}). There are other ways one can bootstrap the label distribution, including adopting sampling schemes that model long-range dependence in the label sequences. For a review of bootstrap methods for dependent data, see for example \cite{buhlmann2002bootstraps}, \cite{horowitz2003bootstrap}, and \cite{kreiss2011bootstrap}.

\section{Conclusions}\label{sec:discussion}

We describe a statistical framework for analyzing the relationship between complex high-dimensional data $\series{\X_t}$ and labels $\series{Y_t}$. For \DID\ sequence data $\series{(\S_t, Y_t)}$, where $\S_{<t}=\{\X_{t-T},\X_{t-T+1}, \ldots,\X_{t}\}$ and $T>0$, we propose a two-sample test (Equation~\ref{eqn:new_null} and Algorithm \ref{alg:distributional-differences}) with minimal assumptions beyond stationarity. The test relies on two simple key ideas: (i) a test statistic based on the posterior difference $\P(Y=1|\S)-\P(\S)$, which we estimate using a machine learning algorithm suitable for the data at hand (empirical results indicate that the test power depends on the quality of the regression estimate; Section \ref{sec:synthetic} and Figure \ref{fig:setting-C-power}); and (ii) a bootstrap test, where we estimate the marginal distribution of $\{Y_t\}_{t\ge0}$ (consistency guarantees asymptotic validity; Theorem \ref{thm:convergence_test_stat}).  Our framework provides interpretable diagnostics in local posterior differences (Section \ref{sec:results_TC}) and can be extended to include longer-range dependence structures (Section \ref{sec:boot}), multiple data sources (Section \ref{sec:discussion_multivar}), and potential confounding variables (Section \ref{sec:confounding}).

\subsection{TC Results}\label{sec:TC_discussion}
We detect a distributional difference between sequences leading up to RI versus not-RI events in both the North Atlantic and eastern North Pacific basins ($p<0.01$). Local posterior differences for RI-NAL and RI-ENP indicate that specific types of convection --- deep and deepening core convection --- are present both before and during RI (Figures~\ref{fig:basinplot}, \ref{fig:case_studies}). Furthermore, we observe that particular convective structures are necessary for RI (Figure~\ref{fig:ri-study}) and thus useful indicators of future RI, but they are not sufficient to trigger RI on their own (Figure~\ref{fig:case_studies}, ii), as the TC environment (e.g., vertical wind shear and ocean heat content) must also support intensification. Thus, while our current results have apparent value for RI forecasting, an analysis of structural trajectories alongside environmental factors such as vertical wind shear promises better understanding of RI and may improve analysis of RW events as well.

When posing the same question regarding RW versus not-RW events, we do not detect a difference in the NAL basin ($p=0.18$); this is expected, as RW is more likely to be driven by a variety of internal and external factors not well captured by convective distribution whereas RI generally \emph{requires} TC convective structure to be capable of sustaining rapid energy uptake. However, the ENP basin is characterized by a narrow spatial region of conditions favorable to TC development such as warm ocean waters and moist air; this homogeneity leads to a significant signal for RW in the ENP basin ($p<0.01$).

\subsection{Broader Impact}\label{sec:broader}
While we apply our methods to meteorology, the proposed statistical framework is applicable to any labeled \DID\ sequence data. Labeled video or other sequence data are common in automation, medical monitoring, and multiple domains in the physical sciences; in many of these areas, the ability to identify high-risk patterns in spatio-temporal data could prove transformative. The flexibility of our framework in admitting a fusion of multiple data sources and adjustment for other variables is also vital to analyzing the complex systems in the environmental and physical sciences and many other domains.\\

%

\section*{Supplemental text} The online supplement contains a description of our CNN regressor in Section \ref{sec:results_TC}, our algorithm for labeling RI and RW events, and a proof of Theorem \ref{thm:convergence_test_stat}.

\section*{Data and Code Availability} All data used are publicly available: HURDAT2 at \url{https://www.nhc.noaa.gov/data/#hurdat} and MERGIR at \url{https://disc.gsfc.nasa.gov/datasets/GPM_MERGIR_1/summary}, which can both be accessed via provided code. All code used to produce the results in this paper is available at \url{https://github.com/ihmcneely/ORB2sample} and as supplementary material.


\bibliographystyle{imsart-nameyear} 
\bibliography{bibliography}       

\clearpage

\appendix
\section{Estimating ${m}_\text{seq}$ and ${m}_\text{post}$}\label{app:mseq_and_mpost}
\textbf{Estimating $m_\text{seq}$:} RI/RW events extending beyond 48 consecutive hours are extremely rare ($\approx5\%$ of events). As such, we estimate $m_\text{seq}$ with a Markov chain of length $k=8$ \emph{at the 6-hour synoptic resolution}. When drawing labels $\series{\Ytilde}$ for a storm, we draw from the Markov chain at synoptic times (0000 UTC, 0600 UTC, 1200 UTC, and 1800 UTC), then fill to the 30-minute resolution as described in Section 3 ($\Ytilde_t=1$ between \emph{consecutive} synoptic $\Ytilde_t=1$). This approximates the distribution of RI event lengths well up 48 hours with a slightly heavy tail beyond 48 hours. We note that each TC requires a separate sequence of draws from $m_\text{seq}$; we thus initialize and burn-in the MC separately for each TC (see Algorithm 1 in main paper for details).

\textbf{Estimating $m_\text{post}$:} Each radial profile is sampled at a 5-km resolution from 0-5km out to 395-400km (for 80 samples in the radius dimension). Our model architecture is identical for all four combinations of basin (NAL/ENP) and event type (RI/RW). Our CNN consists of 3 Convolution-Maximum Pooling layers (8 channels, 8 channels, and 4 channels respectively; convolution size $3\times3$, maximum pooling size $2\times2$) followed by a fully connected layer and a sigmoid output layer. See Table \ref{table:CNN} for the full architecture.

When we initially estimate $m_\text{post}$, we train all parameters. For each subsequent fitting of $\widehat{m}_\text{post}^{(b)}$, we freeze the feature layers (i.e. layers 1, 4, and 7), only re-initializing the fully connected and output layers (i.e. the ``regression'' layers, 11 and 13) for retraining. For each of the four basin-event type combinations, we train 101 CNN models ($B=100$). $\widehat{m}_\text{post}$ achieves approximately 75\%-80\% balanced accuracy on held-out test data for RI in each basin, and approximately 60\%-65\% for RW, while $\widehat{m}_\text{post}^{(b)}$ (trained on random labels) only has $\approx50\%$ balanced accuracy on held-out test data in all cases.

\begin{table}[h]
	\ra{1.2}
	\centering
	\begin{tabular}{@{}rlr@{}}\toprule
		Layer (Type-ID)      &         Output Shape     &    Parameters\\
		\hline
		Conv2d-1:        &    [N, 8, 50, 82]         &      80\\
		ReLU-2:        &    [N, 8, 50, 82]         &      0\\
		MaxPool2d-3:        &    [N, 8, 25, 41]         &      0\\
		Conv2d-4:        &    [N, 8, 27, 43]         &    584\\
		ReLU-5:        &    [N, 8, 27, 43]         &      0\\
		MaxPool2d-6:        &    [N, 8, 13, 21]         &      0\\
		Conv2d-7:        &    [N, 4, 15, 23]         &    292\\
		ReLU-8:        &    [N, 4, 15, 23]         &      0\\
		MaxPool2d-9:        &    [N, 4,  7, 11]        &      0\\
		Dropout-10:        &          [N, 308]         &      0\\
		Linear-11:        &          [N, 256]         & 79,104\\
		ReLU-12:        &          [N, 256]         &      0\\
		Linear-13:        &            [N, 2]         &    514\\
		Sigmoid-14:        &            [N, 2]         &      0\\
		\hline
		Total params: & & \textbf{80,574}\\
		\bottomrule
	\end{tabular}
	\caption{Parameters of the CNN used to estimate $m_\text{post}$. The model consists of three convolutional blocks (layers 1-9). Dropout (layer 10) with probability 50\% is used to regularize the model.}
	\label{table:CNN}
\end{table}

\section{Rapid Intensity Change}\label{app:labeling}
Operationally, RI/RW events are defined by a threshold \emph{and} a time. For example, the rapid intensification index (RII; \cite{kaplan2015rii}) provides the probability of intensification $c$ in a given time period; this includes 24-hour periods such as 25 kt in 24 hr as well as shorter (or longer) periods such as 20 kt in 12 hr (or 55 kt in 48 hr). In lieu of selecting a discretization of time in this manner, we present an algorithm for identifying RI/RW events of any length with a daily rate of intensification of at least $c$ per 24 hours (Algorithm \ref{alg:rapid}). By admitting events both shorter and longer than 24 hours, we prevent the overlap of RI and RW events. See Figure \ref{fig:RI-ident} for an example. The operational (25 kt in 24 hr) definition for RI and RW results in significant overlap between RI and RW events; for example, the transition from RI to RW  on July 15 for Hurricane Carlos [2009] sees an 18-hour overlap between operationally defined RI and RW events. Our labeling procedure restricts the set of $Y=1$ times for each event by trimming non-intensifying (or non-weakening, in the case of RW) times from the beginning and end of each event, while also admitting sufficiently rapid intensification (or weakening) in shorter time periods than 24 hours.

In words, Algorithm \ref{alg:rapid} labels events as such:
\begin{enumerate}
	\item Identify 24-hour periods containing intensification of at least $c$.
	\item Trim the beginning and end of each identified 24-hour RI period such that they begin and end with 6-hour intensification.
	\item Label $Y_t=1$ if any trimmed RI period contains $t$.
\end{enumerate}
RW is identified symmetrically by reversing the input sequence (intensities) in time prior to labelling. By virtue of step 2, RI and RW events can only share endpoints, but cannot overlap.

\begin{figure}[!h]
	\centering
	\includegraphics[width=.7\textwidth]{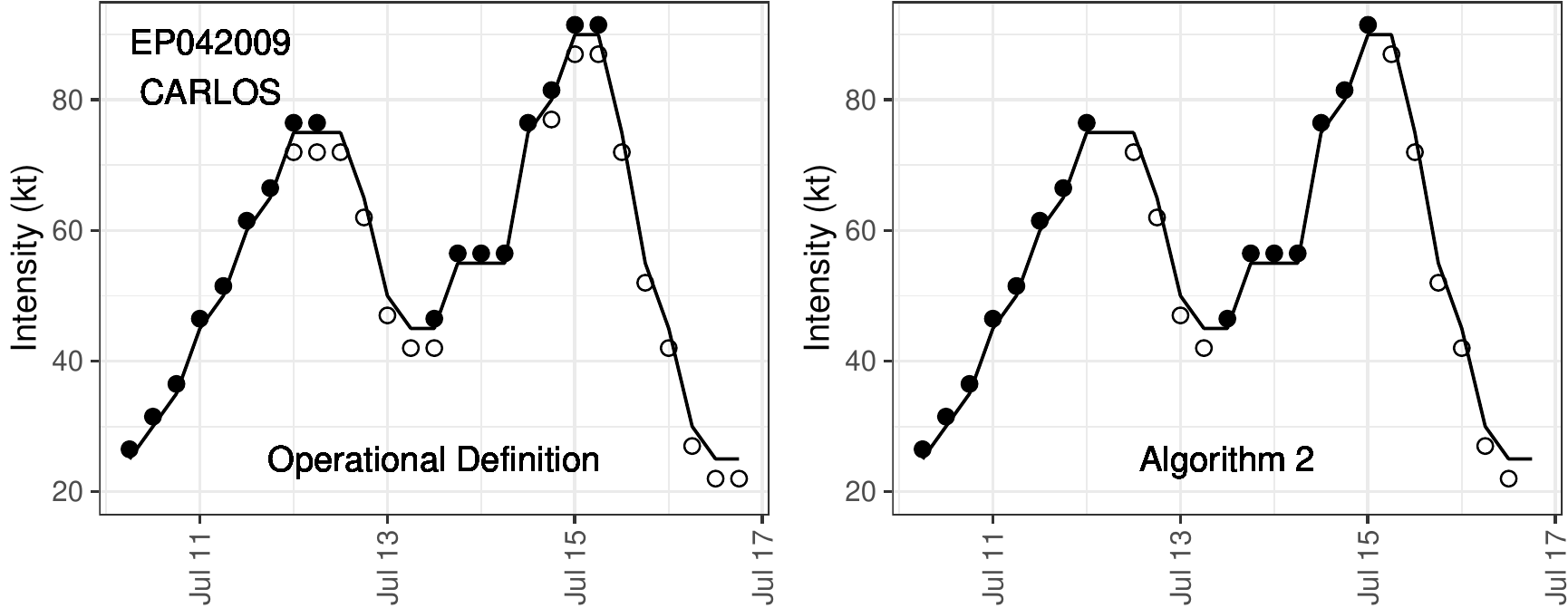}
	\caption{Comparison of RI/RW labeling via 24-hour operational windows (left) versus via Algorithm \ref{alg:rapid} (right). Algorithm \ref{alg:rapid} restricts event labels to only intensifying/weakening times as appropriate, resulting in more logical, non-overlapping labels for the endpoints (e.g., July 15 in the above figure).}
	\label{fig:RI-ident}
\end{figure}

\begin{algorithm}[!h]
	\SetAlgoLined
	\SetKw{req}{Require:}
	\SetKw{output}{Output:} 
	\req Sequence of contiguous intensity (i.e. maximum wind speed) observations $\{W_i\}_{t=1}^T$ for a storm and an intensity change threshold $c$. \\
	\textbf{Initialize:} $\mathbf{Z}=\mathbf{0}_{T,T}$, $Y=\mathbf{0}_T$~\\
	\For{$t$ in $1:T-1$}{
		$\Delta_{t,1}=W_{t+1}-W_t$; \tcp{\texttt{lead-1 intensity change}}~\\ 
		\If{$t\le T-4$}{
			$\Delta_{t,4}=\max(W_t,\dots,W_{t+4})-W_t$; \tcp{\texttt{lead-4 intensity change}}
		}
	}
	$\mathcal{A}=\{t:\Delta_{t,4}\ge c\}$; \tcp{\texttt{24-hour windows containing RI}}
	
	$\mathcal{B}=\{t:\Delta_{t,1}>0\}$; \tcp{\texttt{6-hour windows containing intensification}}
	
	\For{$t$ in $1:T-4$}{
		\If(\tcp{\texttt{label 24-hour RI window}}){$t\in\mathcal{A}$}{$Z_{t,t},\dots,Z_{t,t+4}=1$;\\
			$h=4$;\\
			\While{$h\ge1$ and $t+h-1\in\mathcal{B}$}{
				$Z_{t,t+h}=0$;\tcp{\texttt{trim non-intensification from end of event}}~\\
				$h=h-1$;
			}
			$h=1$;\\
			\While{$h<4$ and $t+h\in\mathcal{B}$}{
				$Z_{t,t+h}=0$;\tcp{\texttt{trim non-intensification from start of event}}~\\
				$h=h+1$;
			}
		}
	}
	\For{$t$ in $1:T$}{
		$Y_t=\max(Z_{t,1},\dots,Z_{t,T})$; \tcp{\texttt{points only need to be valid for one start}}
	}
	\output $\{Y_i\}_{t=1}^T$\\
	\textbf{Note:} {\em The above algorithm identifies rapid intensification. To identify rapid weakening instead, reverse the input sequence $\{W_i\}_{t=1}^T$ at initialization, then reverse the output sequence $\{Y_i\}_{t=1}^T$ at output.}
	\caption{Identifying observations that fall within rapid intensification events.}
	\label{alg:rapid}
\end{algorithm}

\clearpage

\section{Proofs}\label{app:proofs}

\begin{proof}[Proof of Theorem 1]
	Let $F_{G_{\widehat \p_{t_2}}}$ be the cumulative distribution function of $\{Y^0_t\}_{t \in \mathcal{T}_{\text{1}}}$. Then,
	under the null hypothesis,
	\begin{align*}
	\P\left( \S_{<t} \leq \s_{<t}, Y^0_t  \leq y_t, \  \forall t \in  \mathcal{T}_{1} \right)&=
	\P\left( \S_{<t} \leq \s_{<t},  \  \forall t \in  \mathcal{T}_{1} \right)\P\left( Y^0_t \leq y_t,  \  \forall t \in  \mathcal{T}_{1} \right) \notag  && \text{By construction} \\
	&=   \P\left( \S_{<t} \leq \s_{<t},  \  \forall t \in  \mathcal{T}_{1} \right)F_{G_{\widehat \p_{t_2}}}\left( y_t,  \  \forall t \in  \mathcal{T}_{1} \right) && \text{By definition}\\ 
	&\xrightarrow[t_2 \longrightarrow\infty]{\enskip  \enskip}
	\P\left( \S_{<t} \leq \s_{<t},  \  \forall t \in  \mathcal{T}_{1} \right)F_{G^*}\left( y_t,  \  \forall t \in  \mathcal{T}_{1} \right)&& \text{Assumption 4}\\ 
	&=\P\left( \S_{<t} \leq \s_{<t},  \  \forall t \in  \mathcal{T}_{1} \right)\P\left( Y_t \leq y_t,  \  \forall t \in  \mathcal{T}_{1} \right) && \text{By definition}\\
	&=\P\left( \S_{<t} \leq \s_{<t},Y_t \leq y_t,   \  \forall t \in  \mathcal{T}_{1} \right). && \text{Assumption 2}
	\end{align*}
	It follows that
	$$ \mathcal{D}^{t_2}_0 \xrightarrow[t_2 \longrightarrow\infty]{\enskip \mbox{Dist} \enskip}  \mathcal{D}.$$
	The conclusion of the theorem follows from the continuous mapping theorem and Assumption 3.
\end{proof}

\begin{proof}[Proof of Corollary 1]
	By the law of large numbers, for every realization $D$ of $\mathcal{D}$,
	$$\widehat p^{\ t_2}_B(D)  \xrightarrow[B \longrightarrow\infty]{\enskip \mbox{a.s.} \enskip}  \ \P\left(\lambda(\mathcal{D}^{t_2}_0)> \lambda({D})\right).$$
	Now, Theorem 1 implies that
	\begin{align*}
	\lim_{t_2 \longrightarrow \infty} \P\left(\lambda(\mathcal{D}^{t_2}_0)> \lambda({D})\right)&=1-\lim_{t_2 \longrightarrow \infty} \P\left(\lambda(\mathcal{D}^{t_2}_0)\leq  \lambda({D})\right)\\
	&=1-\lim_{t_2 \longrightarrow \infty}F_{\lambda(\mathcal{D}^{t_2}_0)}\left( \lambda({D})\right)=1-F_{\lambda(\mathcal{D})}\left( \lambda({D})\right).
	\end{align*}
	Conclude that 
	$$\lim_{t_2 \longrightarrow \infty}\lim_{B \longrightarrow \infty} \widehat p^{\ t_2}_B(\mathcal{D})=1-F_{\lambda(\mathcal{D})}\left( \lambda(\mathcal{D})\right)$$
	almost surely, and thus this convergence also holds in distribution.
	The conclusion follows from the fact that $F_{\lambda(\mathcal{D})}\left( \lambda(\mathcal{D})\right)$ is a uniform random variable, and therefore so is $1-F_{\lambda(\mathcal\mathcal{D})}\left( \lambda({D})\right)$.
\end{proof}

\section{Local p-values}\label{app:local_p_vals}

In this section we show that LPDs can be used  to
test the local null hypothesis
\begin{align}
\label{eq:local_null_ball}
H_0^{\epsilon}(\s): \mathbb{P} (Y=1 | \mathbf{S}=\mathbf{s}') = \mathbb{P} (Y=1)
\text{ for all } \mathbf{s}' \in B(\s;\epsilon),
\end{align}
where $B$ is a ball of radius $\epsilon$ centered at $\s$.

This local test needs further assumptions to asymptotically control 
type I error probabilities, which we describe in what follows.
\begin{Assumption}[\textbf{Local regression estimator}]
	\label{assump:local_reg}
	$\widehat m_{\text{post}}$ only uses the observations in $\mathcal{D}$ such that $\S_{<t} \in B(\s;\epsilon)$.
\end{Assumption}

Assumption \ref{assump:local_reg} holds for regression estimators that are based on partitions, such as tree-based estimators (e.g., random forests, boosting methods) or smoothing kernel estimators with kernels with bounded support.

\begin{Assumption}[\textbf{Conditional independence}]
	\label{assump:DAG_B}
	$\{(\S_{<t},Y_t)\}_{t \geq 0}$ satisfies the DAG of Setting B (Figure 2).
\end{Assumption}

\begin{thm}
	\label{thm:local_pvalues}
	Assume that the measure $Q$ in the test statistic from Equation 8 places all of its mass on $B(\s;\epsilon)$. Also assume that Assumptions 1, 2,  \ref{assump:local_reg} and \ref{assump:DAG_B} hold. Then, under the local null hypothesis in Equation \ref{eq:local_null_ball}, if the marginal distribution is estimated using the estimator of Example 1, then, for any $0<\alpha<1$
	$$\lim_{t_2 \longrightarrow \infty}\lim_{B \longrightarrow \infty} \P\left( \widehat p^{\ t_2}_B(\mathcal{D})\leq \alpha \right) =\alpha.$$
\end{thm}

\begin{proof}

	Let $\mathcal{I}=\{t \in \mathcal{T}_1: \S_{<t} \in B(s;\epsilon) \}$. Under $H_0^{\epsilon}(\s)$,
	\begin{align*}
	\P( \S_{<t} \leq \s_{<t} \ &\forall t \in  \mathcal{T}_{1} \mbox{ and } Y^0_t  \leq y_t \ \forall t \in  \mathcal{I}  |\mathcal{I}  ) \\
	&=
	\P\left( \S_{<t} \leq \s_{<t}  \  \forall t \in  \mathcal{T}_{1} \right|\mathcal{I})\P\left( Y^0_t \leq y_t  \  \forall t \in  \mathcal{I}  |\mathcal{I}\right) \notag  && \text{By construction} \\
	&\xrightarrow[t_2 \longrightarrow\infty]{\enskip  \enskip}
	\P\left( \S_{<t} \leq \s_{<t}  \  \forall t \in  \mathcal{T}_{1}  | \mathcal{I}  \right)\P\left( Y_t \leq y_t  \  \forall t \in  \mathcal{I} | \mathcal{I} \right) && \text{Example 1}\\
	&=\P( \S_{<t} \leq \s_{<t} \ \forall t \in  \mathcal{T}_{1} \mbox{ and } Y_t  \leq y_t \ \forall t \in  \mathcal{I}  |\mathcal{I}  ). && \text{Assumption \ref{assump:DAG_B}}
	\end{align*}
	
	The result follows from noticing that, under Assumption \ref{assump:local_reg}, $\lambda(\mathcal{D})$  depends on $\mathcal{D}$ only through  $\{\S_{<t}\}_{t \in  \mathcal{T}_{1}}$ and $\{ (\S_{<t}, Y_t) \}_{t \in \mathcal{I}}$ and by repeating the arguments of the proofs of Theorem 1 and Corollary 1.

\end{proof}


\end{document}